\documentclass[12pt]{amsart}
\usepackage{geometry} 
\usepackage{graphicx}
\usepackage{amsfonts}
\usepackage{amsthm}
\usepackage{amsmath}
\usepackage{amssymb}
\usepackage{url}
\usepackage[all]{xy}
\usepackage{enumerate}
\usepackage{color}
\usepackage{multirow}
\usepackage{array}
\usepackage[normalem]{ulem}
\usepackage{tikz}
\usepackage{xspace}
\usepackage{subfig}
\usetikzlibrary{decorations.markings}

\usepackage{enumitem}
\usetikzlibrary{arrows,decorations.markings}
\newtheorem{theorem}[subsubsection]{Theorem}
\newtheorem{corollary}[subsubsection]{Corollary}
\newtheorem{lemma}[subsubsection]{Lemma}
\newtheorem{proposition}[subsubsection]{Proposition}
\newtheorem*{Proposition}{Proposition A}
\theoremstyle{definition}
\newtheorem{example}[subsubsection]{Example}
\newtheorem{definition}[subsubsection]{Definition}

\theoremstyle{remark}
\newtheorem{remark}[subsubsection]{Remark}

\newcommand{\Lk}{\mathop{\rm Lk}}

\newcommand{\Hom}{\operatorname{Hom}}

\newcommand{\gen}[1]{\langle #1\rangle}

\newcommand{\cls}[1]{\overline{#1}}

\newcommand{\ins}{\subseteq}
\newcommand{\di}{\partial}

\tikzset{
  big arrow/.style={
    decoration={markings,mark=at position 1 with {\arrow[scale=2.5]{>}}},
    postaction={decorate},
    shorten >=0.4pt}}
\tikzset{middlearrow/.style={
        decoration={markings,
            mark= at position 0.5 with {\arrow{#1}} ,
        },
        postaction={decorate}
    }
}

\title[Geometrization of Adinkras II]{Geometrization of $N$-extended 1-Dimensional Supersymmetry Algebras, II.}

\author{Charles Doran}
\address[Charles Doran, Jordan Kostiuk]{Department of Mathematical and Statistical Scineces, University of Alberta, Edmonton, AB T6G 2G1, Canada}
\email{charles.doran@ualberta.ca, jkostiuk@ualberta.ca}

\author{Kevin Iga}
\address[Kevin Iga]{Natural Science Division, Pepperdine University, Malibu, CA 90263, USA}
\email{kiga@pepperdine.edu}

\author{Jordan Kostiuk}

\author{Stefan M\'{e}ndez-Diez}
\address[Stefan M\'{e}ndez-Diez]{Mathematics Program, Bard College, Annandale-on-Hudson, NY 12504-5000, USA}
\email{smendezdiez@bard.edu}

\begin{document}

\begin{abstract}
The problem of classifying off-shell representations of the $N$-extended one-dimensional super Poincar\'{e} algebra is closely related to the study of a class of decorated $N$-regular, $N$-edge colored bipartite graphs known as Adinkras. In previous work we canonically embedded these graphs into explicitly uniformized Riemann surfaces via the ``dessins d'enfants'' construction of Grothendieck.  The Adinkra graphs carry two additional structures: a selection of dashed edges and an assignment of integral heights to the vertices.  In this paper, we complete the passage from algebra, through discrete structures, to geometry.  We show that the dashings correspond to special spin structures on the Riemann surface, defining thereby super Riemann surfaces.  Height assignments determine discrete Morse functions, from which we produce a set of {\em Morse divisors} which capture the topological properties of the height assignments.  
\end{abstract}

\maketitle

\tableofcontents

\section{Introduction}

Graphs known as Adinkras were proposed by Faux and Gates in \cite{Faux:2004} as a fruitful way to investigate off-shell representations of the super Poincar\'{e} algebra. Adinkras are graphs whose vertices represent the particles in a supermultiplet and whose edges correspond to the supersymmetry generators. In combinatorial terms, Adinkras are $N$-regular, edge $N$-colored bipartite graphs with signs assigned to the edges, and heights assigned to the vertices, subject to certain conditions. Details can be found in Section 2 of \cite{Doran:2015}.
\footnote{See \cite{Doran2017} for the algebraic formulation of this classification problem in terms of filtered Clifford supermodules.}

It is useful to think of an Adinkra as consisting of a {\em chromotopology}, which captures the underlying bipartite graph with its $N$-coloring, together with two more compatible structures: an {\em odd dashing}, which marks each edge with a sign, and a {\em height assignment}, which labels each of the vertices with an integer. A complete characterization of chromotopologies was achieved in \cite{Doran:2011}. For each $N$, there is a natural chromotopology on the Hamming cube $[0,1]^N$, with vertices labeled by elements of $\mathbf{F}_2^N$. The $1$-skeleton of the Hamming cube serves as a ``universal cover'' for arbitrary chromotopologies, the covering map being realized by taking cosets with respect to doubly even binary linear error correcting codes $C \subseteq \mathbf{F}_2^N$. In particular, to each Adinkra chromotopology, there is an associated doubly-even code $C$ realizing it as the quotient of the Hamming cube. 

In Part I of this paper, we constructed a Riemann surface out of the underlying chromotopology of an Adinkra $A$ in such a way that the Adinkra graph is the $1$-skeleton of the Riemann surface. 
This construction was described explicitly in three different ways. 
The first description used Grothendieck's theory of  \emph{dessins d'enfants}, which associates a canonical Riemann surface, presented as a cover of the Riemann sphere branched over three points, to a ribbon graph structure---a bipartite graph, together with a cyclic ordering of the edges incident to each vertex. 
In order to place a ribbon graph structure on an Adinkra, we fixed a \emph{rainbow}, which is a cyclic ordering of the $N$ edge-colors.
With this choice made, we oriented the edges incident to a white vertex in agreement with the rainbow, while orienting the edges incident to a black vertex in the opposite order of the rainbow.  
This choice of orientations gave rise to a Riemann surface that \emph{minimally} embedds the Adinkra (with respect to genus). 
Indeed, for the $N$-cube, this is classical and follows from  Euler's classical polyhedron formula \cite{Beinekea}, and the result follows in general with a slight modification of this classical argument. 
By design, the structure map factored through a degree-$N$ mapping of the Riemann sphere to itself; the image of the Adinkra on this sphere is a graph with one white vertex, one black vertex, and one edge of each color joining the two. 
This allowed us to consider all of the Riemann surfaces constructed from Adinkras as branched covers of this ``beachball'', which we denoted by $B_N$. 
One advantage of this approach was the following: the deck-transformation group from the hypercube surface to the beach ball is the maximally even code of length $N$. 
Adinkras corresponding to doubly-even codes are then the free quotients by the doubly-even code viewed as a subgroup of deck transformations. 
Not only is this  an attractive feature from the point of view of covering space (i.e., the covers are unbranched), but this also reflects the importance/elegance of the connection between Adinkras and doubly-even codes.

The second description of our construction is given in terms of  Fuchsian uniformizations.
All of the Riemann surfaces described above correspond, via the uniformization theorem, to torsion-free subgroups of the $(2,N,N)$-triangle group. 
Thus, for $N>4$, our Riemann surfaces can be presented as quotients of the upper half-plane by a subgroup of M\"{o}bius transformations. 
In our work, we \emph{explicitly} determined the Fuchsian subgroups that gave rise to our surfaces in terms of the doubly-even code that is associated to the Adinkra.

Finally, we described the hypercube Riemann surface as a complete intersection of quadrics in projective space, and the intermediate surfaces as quotients by deck transformations. 
This presentation yielded a description of the deck transformations over the beach ball as sign changes on coordinates with the centers of the quadrilateral faces corresponding to the coordinate hyperplane sections.
This model is visibly defined over a (real subfield of a) cyclotomic field, but it was shown using the theory of Galois descent that the curves are in fact definable over the rational numbers. 
A priori, the curves were definable over $\cls{\mathbf{Q}}$ because of the Belyi map, but in fact these curves are much more special --- this is a striking reflection of the symmetry appearing in supersymmetry. 

It is interesting to note that the Riemann surfaces constructed from Adinkra graphs share many of the properties of the most special Riemann surfaces studied in arithmetic geometry.  For example, the famous Fricke-Macbeath curve of genus seven shares much in common with with the curves associated to the $N=7,k=3$ codes.  Such curves are genus seven covers of the seven-pointed sphere, are defined over $\mathbf{Q}$, have monodromy group isomorphic to $(\mathbf{Z}/2\mathbf{Z})^3$, and have Jacobian abelian varieties that split completely up to isogeny into elliptic curve factors.  The methods of computing the isogenous decomposition for Adinkra Riemann surfaces developed in this paper, however, permit one to show that these two Riemann surfaces are nevertheless {\em not} isomorphic.

Denote by $\mathcal{R}(A)$ the set of all pairs $(A,r)$ where $A$ is an Adinkra and $r$ is a rainbow. 
If $\mathcal{A}=(A,r)\in\mathcal{R}(A)$, the Riemann surface constructed in \cite{Doran:2015} will be denoted by $X_\mathcal{A}$; we often denote the Riemann surface associated to the $N$-cube simply by $X_N$, and we will make the rainbow explicit if needed. 
One clarification also needs to be made.  
We had claimed that $R$-symmetry produces isomorphic Riemann surfaces; that is, changing the fixed rainbow would produce isomorphic surfaces. 
Unfortunately, this is not true in general, although it \emph{is} true for the $N$-cube, and the proof found in \cite{Doran:2015} is correct. 
Certain $R$-symmetries will yield isomorphic Riemann surfaces. 
One obvious such $R$-symmetry is the rotation $R$-symmetry, where each edge is given the next color in the rainbow. 
That this produces an isomorphic Riemann surface can best be seen by observing that such a rainbow change is induced by any deck transformation that lifts the rotation morphism on the beachball $B_N$, as described in \cite{Doran:2015}. 
In addition to the rotation symmetry, it is safe to act on the rainbow via permutations that stabilize the associated doubly-even code. 
That is, such permutations \emph{will} give rise to isomorphic Riemann surfaces. 
A proof of this fact is in Appendix A. 
We propose that at this point in the construction, the invariant object in question should be a set of Belyi curves $X_{\mathcal{A}}$, indexed by the set of possible rainbows.

In the present paper we complete the passage from supersymmetry algebras, through discrete structures (Adinkras), to {\em geometry} by naturally interpreting both of the remaining Adinkra structures in geometric terms.  

The Adinkra graph on the Riemann surface defines a {\em dimer model}, and using work of Cimasoni-Reshetikhin \cite{Cimasoni:2007} and the classification of odd-dashings via cubical cohomology in \cite{Doran2012b} we show that the odd-dashings correspond to special spin structures on the Riemann surface, and hence to {\em super Riemann surfaces}.  

From the height assignments on the Adinkra graphs we construct {\em discrete Morse functions} in the sense of Banchoff \cite{Banchoff:1970} on the topological surface.  
Using this description, we produce a set of {\em Morse divisors} on the Riemann surface which capture the topological properties of the height assignments.  
In Theorem \ref{jacobiandecom} we show that the Jacobians of these Riemann surfaces are always isogenous to the product of Jacobians of all possible intermediate hyperelliptic subcovers of $B_N$.
This partial decomposition allows us to see the non-trivial action of $R$-symmetry.

We conclude with some first comments on the lessons learned from considering the full geometric package coming from supersymmetry via Adinkras.

We would like to thank S.J. Gates, Jr., P. Green, T. H\"{u}bsch, and G. Landweber for extended discussions and useful suggestions while writing this paper. The argument in Appendix B is a modification of P. Green's proof recast in the language of Kani-Rosen needed for our purposes. CD and JK acknowledge support from the Natural Sciences and Engineering Research Council of Canada, the Pacific Institute for the Mathematical Sciences, and the Campobassi  Professorship at the University of Maryland.  All of the authors acknowledge the support of the National Science Foundation Grant \#1602991, which provided for a workshop in supersymmetry that allowed us to meet and work on this project together.

\section{From Odd Dashings to Spin Structure}

In this section, we will show how odd dashings on an Adinkra correspond to spin structures on the associated Riemann surfaces. 
Whereas in the first part of this paper our focus was on the construction of a Riemann surface $X_\mathcal{A}$ from the chromotopology of the Adinkra graph $A$, we now make critical use of both the embedding of $A$ in $X_\mathcal{A}$ and of the odd dashing structure on $A$. 
Cimasoni and Reshetikhin have shown that Kasteleyn orientations on a graph embedded in a Riemann surface are in one-to-one correspondence with spin structures once a perfect matching, or \emph{dimer configuration}, of the edges has been fixed \cite{Cimasoni:2007}. 
The Adinkra $A$ embedded in $X_\mathcal{A}$ comes with $N$ natural choices for dimer configurations by considering the subset of edges of a fixed color, thereby allowing us to use the theory of Kasteleyn orientations on graphs embedded in a Riemann surface developed in \cite{Cimasoni:2007}.

Before constructing spin structures from odd dashings, we review the cubical complex associated to an Adinkra that was used to classify the set of odd dashings on an Adinkra  in \cite{Doran2012b}. 
There are striking similarities to the structure of the set of all such odd dashing, as well as the obstruction to their existence when considering graphs corresponding to even codes, not just Adinkras.
These similarities --- together with the fact that the Riemann surface with embedded graph from the first part of this paper provides exactly what Cimasoni and Reshetikhin use for their construction --- inspired us to associate spin structures to odd dashings in the first place.

\subsection{Review of the Cubical Complex}
 Let $I$ be the unit interval, and $I^N$ the $N$-dimensional hypercube.
 The hypercube $I^N$ has a natural CW-complex structure with $3^N$ cells corresponding to the elements of $\{0,1,*\}^N$, where $*$ is a formal symbol for the unit interval. The dimension of the cell corresponding to $(x_1,\dots, x_N)\in\{0,1,*\}^N$ is the number of $x_j$ for which $x_j=*$.  
The Adinkra $A_N$ corresponding to the hypercube is the $1$-skeleton of $I^N$ in this cell complex. The white and black vertices correspond to the $0$-cells with an even, respectively, odd number of $1$'s. An edge is a $1$-cell $(x_1,\dots, x_N)$, and is colored $j$ for the unique entry for which $x_j=*$.

We will be interested in cubical homology with both $\mathbf{Z}$ and $\mathbf{Z}/2\mathbf{Z}$ as coefficient rings.  
Let $C_k(I^N)$ be the free abelian group (or $\mathbf{Z}/2\mathbf{Z}$-module) generated by the $k$-dimensional cells.
Given a cell $(x_1,\cdots,x_N) \in\{0,1,*\}^N$, an integer $j \in\{0,\cdots,N\}$, and $\alpha \in
\{0,1,*\}$,  the replacement operator  $\textrm{Repl}_{j,\alpha}$ is defined as $$\textrm{Repl}_{j,\alpha}(x_1,\cdots,x_N) = (x_1,\cdots,x_{j-1},\alpha,x_{j+1},\cdots,x_N).$$
 These replacement operators are used in \cite{Doran2012b} to define boundary operators on the associated chain complex.

For each $1\leq j\leq N$, the reflection map $\rho_j\colon\{0,1,*\}^N\to \{0,1,*\}^N$ is defined by 
$$\rho_j(x_1,\dots, x_N)=\left\{\begin{array}{ll}(x_1,\dots, x_N)&x_j=*\\
\textrm{Repl}_{j,1-x_j}(x_1,\dots, x_N)&x_j\neq *.	
\end{array}\right.$$
Now let $C$ be a binary code of length $N$, and for each $(t_1,\dots, t_N)\in C$, set
$$\rho_{(t_1,\dots, t_N)}=\rho_1^{t_1}\cdots\rho_N^{t_N}.$$
In this way, the code $C$ acts naturally on the cells of $I^N$ and, in particular, on its $1$-skeleton. Taking for $C$ the doubly-even code associated to $A$, the quotient of the $1$-skeleton by $C$ is identified with $A$. Since the boundary operators commute with the above reflection operators, the complex of cells in the quotient has the structure of a CW-complex, which we call the \emph{cubical complex} of $A$.  The homology groups of the cubical complex will be denoted by $H_k(A,-)$ and called the cubical homology groups; the cubical cohomology groups are defined dually.  
\begin{proposition}
	Let $A$ be an Adinkra, and let $C$ be the associated doubly even binary linear code. Then we have:
	\begin{enumerate}
		\item $H_0(A,\mathbf{Z})\cong\mathbf{Z}$;
		\item $H_1(A,\mathbf{Z})\cong C$;
		\item $H_0(A,\mathbf{Z}/2\mathbf{Z})\cong\mathbf{Z}/2\mathbf{Z}$;
		\item $H_1(A,\mathbf{Z}/2\mathbf{Z})\cong H^1(A,\mathbf{Z}/2\mathbf{Z})\cong C$.
	\end{enumerate}
\end{proposition}

\begin{proof}
	The first and third isomorphisms following from the connectedness of the cubical complex. The calculation of the first homology group follows from the fact that $I^N$ is the universal cover of the cubical complex of $A$ and that the covering group is naturally isomorphic to the code $C$. 
	The details of this argument, as well as the computation of cohomology group can be found in \cite[Theorem 4.1]{Doran2012b}.
\end{proof}

\subsection{Embedding into the Cubical Complex}
While the cubical complexes described above are interesting objects to study, it is not natural to give them a manifold structure. Our construction in \cite{Doran:2015} of course \emph{does} yield a manifold and, in this section, we compare these two objects by defining a topological embedding of the Riemann surface inside the cubical complex.

\begin{proposition}
Let $\mathcal{A}\in\mathcal{R}(A)$, $X_\mathcal{A}$ be the associated Riemann surface and $C(A)$ the cubical complex.  
There is an inclusion of CW-complexes $i\colon X_\mathcal{A}\to C(A)$, identifying $X_\mathcal{A}$ with the subset of the $2$-skeleton of $C(A)$ consisting of all of the $2$-cells associated to adjacent colors in the rainbow associated to $\mathcal{A}$.  
\end{proposition}

\begin{proof}
 Choose a white vertex $w\in A\ins X_\mathcal{A}$ and map it to $(0,\dots, 0)$ in the cubical complex. If $w'$ is another vertex, then $w'=(t_1,\dots, t_N)w$ for some uniquely determined $(t_1,\dots, t_N)\in (\mathbf{Z}/2\mathbf{Z})^N/C$; map $w'$ to $\rho_{(t_1,\dots, t_N)}(0,\dots, 0)$ in the cubical complex. This defines $i$ on the $0$-skeleton. By construction, two vertices that are adjacent by an edge of a given color in $A\ins X_\mathcal{A}$ will map to two vertices that are adjacent via an edge of the same color. This shows how to extend $i$ to the $1$-skeleton. It is injective because there are no double edges, so $i$ is an isomorphism on the $1$-skeleton. Finally, $i$ is extended to the $2$-skeleton by mapping each $j/j+1$ colored face to the corresponding $j/j+1$ colored face in the cubical complex that has the corresponding boundary edges. This defines the inclusion on all of $X_\mathcal{A}$. 	
\end{proof}

The above proposition allows us to identify $X_\mathcal{A}$ (as a topological space) as a subset of the cubical complex. It follows that $X_\mathcal{A}$ inherits the structure of cubical set from the cubical complex for $A$, and we may therefore compute the homology groups of $X_\mathcal{A}$ via the associated chain complex.

\begin{corollary}
The inclusion $i\colon X_\mathcal{A}\to C(A)$ induces an isomorphism of homology in degree $0$, and a surjection in degree $1$. The kernel of
$$H_1(X_\mathcal{A},\mathbf{Z})\to H_1(A,\mathbf{Z})$$ is the subspace $L$ generated by the images of all $2$-colored loops on the Adinkra $A\ins X_\mathcal{A}$.	In particular, using the isomorphism $H_1(A,\mathbf{Z})\cong C$, we have an exact sequence
$$\xymatrix{0\ar[r]&L\ar[r]&H_1(X_\mathcal{A},\mathbf{Z})\ar[r]&C\ar[r]&0}.$$
\end{corollary}

\begin{proof}
The inclusion is an isomorphism on the $1$-skeleton, whence the isomorphisms in degree $0$. Now consider the map $H_1(X_\mathcal{A},\mathbf{Z})\to H_1(A,\mathbf{Z})$. 
		 To see that $i_*$ is surjective in degree $1$, we use the fact that $H_1(A,\mathbf{Z})$ is generated by the cycles $P^c$ for $c\in C$, where $C$ is the code associated to $A$, and $P^c$ is the cycle that starts at some fixed vertex $w$ and traverses the edges of each color corresponding the non-zero entries of $c$ in sequence \cite{Doran2012b}.
	 These cycles lie in $X_\mathcal{A}$, showing that $i$ induces a surjection in homology. To determine the kernel, consider the following commutative diagram coming form the chain complexes. 
	 
	 $$\xymatrix{C_2(A)\ar^{\di_A}[r]&C_1(A)\\
C_2(X_\mathcal{A})\ar[u]\ar^{\di_X}[r]&C_1(X_\mathcal{A})\ar[u]}.$$
Suppose $\xi\in C_1(X_\mathcal{A})$ represents the class $[\xi]\in H_1(X_\mathcal{A},\mathbf{Z})$, 
and that $i_*([\xi])=0\in H_1(A,\mathbf{Z})$. It follows that
 $\xi\in\textrm{image}(\di_A)$. 

On the other hand, the image of $\di_I$ is, by definition, the subspace generated by the images of all $2$-colored loops on $A$, whence the result.
\end{proof}

\subsection{Spin Structures via Cimasoni-Reshetikhin}
Our ultimate goal for the rest of this section is to show that the odd dashing on an Adinkra gives rise to a spin structure on $X_\mathcal{A}$. In order to do this, we use the work of Cimasoni and Reshetikhin \cite{Cimasoni:2007}; in their geometric proof of the Pfaffian formula for closed Riemann surfaces, they obtain a precise relationship between Kasteleyn orientations on a graph with dimer configuration and (discrete) spin structures. We review the necessary details below.  

Let $X$ be a Riemann surface and $\Gamma\ins X$ an embedded graph that induces a cellular decomposition of $X$.

\begin{definition}
	An \emph{orientation} of the edges of  $\Gamma$ is an assignment of a direction to each edge. The orientation $K$ is called \emph{Kasteleyn} if the following parity condition holds: for each face $f$ of $\Gamma$, oriented counterclockwise, there is an odd number of edges of $f$ whose orientation does not agree with that of $K$.
\end{definition}

\begin{definition} Given an orientation $K$ on $\Gamma$, and a vertex $v$, the \emph{vertex switch at $v$ of $K$} is the orientation $K'$ obtained from $K$ by reversing the orientations at each edge incident to $v$. 
	Two Kasteleyn orientations are said to be equivalent if one is obtained from the other by a sequence of vertex switches. The set of Kasteleyn orientations is denoted by $\mathcal{K}$, and  the set of equivalence classes of Kasteleyn orientation is denoted by $\widetilde{\mathcal{K}}$. 
\end{definition}

If we assume that $\Gamma$ is bipartite  the parity condition for a Kasteleyn orientation can be reformulated in terms of the bipartite orientation of $\Gamma$, defined, say, by orienting all edges from the incident white vertex to the incident black vertex. 
Given a face $f$ of $X$, let $\ell$ denote the number of times the orientation $K$ disagrees with the bipartite orientation on the face $f$ and let $2k$ denote the number of edges that bound $f$; the orientation $K$ is Kasteleyn if and only if $k$ and $\ell$ have the same parity. 
In particular, if each face of $X$ is a $4$-gon, then the orientation is Kasteleyn if and only if $K$ disagrees with the bipartite orientation an odd number number of times for each face. 
From here on out we will assume that $\Gamma$ is bipartite and that each face is a $4$-gon. 

Kasteleyn orientations can be characterized algebro-topologically.  Indeed, first note that since the graph $\Gamma$ induces a cellular decomposition of $X$, we can compute its homology groups with $\mathbf{Z}/2\mathbf{Z}$ coefficients using the corresponding CW-complex. Let $C_k(X)$ denote the group of $k$-chains of $X$, and $C^k(X)$ the group of $k$-cochains. A Kasteleyn orientation $K$ of $\Gamma$ can then be thought of as the $1$-cochain $\epsilon_K$ that sends an edge $e$ to $1$, if the orientation of $e$ in $K$ disagrees with the bipartite orientation, and to $0$ otherwise. Let $\omega_2$ denote the $2$-cochain that sends each face to $1$; the parity condition implies that $d\epsilon_K=\omega_2$. Conversely, a $1$-cochain $\epsilon$ such that $d\epsilon=\omega_2$ defines a Kasteleyn orientation by starting with the bipartite orientation, and then reversing the orientation of the edges that map to $1$. 

\begin{proposition}
	 Two Kasteleyn orientations $K$ and $K'$ are related by a series of vertex switches if and only if $\epsilon_K-\epsilon_{K'}=d\zeta$ for some $0$-cochain $\zeta$. If $K$ and $K'$ are any Kasteleyn orientations, then $\epsilon_K-\epsilon_{K'}$ is in fact a cocycle. 
\end{proposition}

\begin{proof}
Let $v\in\Gamma$ be a vertex, and let $\zeta_v$ denote the $0$-cochain that sends $v$ to $1$, and all other vertices to $0$. If $\epsilon_K$ is a Kasteleyn orientation, then it is easy to see that $\epsilon_K+d\zeta_v$ is the orientation obtained as the result of applying the vertex switch at $v$. Conversely, if $\epsilon_K-\epsilon_{K'}=d\zeta$, then $\epsilon_{K'}$ is obtained from $\epsilon_K$ by applying vertex switches at all of the vertices in $\zeta^{-1}(1)$. Finally, it was already observed that $d\epsilon_K=\omega_2$ for any Kasteleyn orientation, from which it follows that the difference of any two Kasteleyn orientations is a cocycle.
\end{proof}

\begin{definition}
	Let $L$ be an abelian group. An \emph{affine space over $L$} is a set $M$ equipped with an ``addition'' $+\colon L\times M\to M$ satisfying \begin{itemize}
		\item $\lambda_1+(\lambda_2+m)=(\lambda_1+\lambda_2)+m$ for $\lambda_i\in L,m\in M$;
		\item given any $m_1,m_2\in K$, there is a unique $\lambda\in L$ such that $m_2=\lambda+m_1$. 
	\end{itemize}
	
An \emph{affine subspace} of $M$ is an affine space $N\ins M$ over $L_1\ins L$. 
A morphism of affine spaces $N,M$ over $L$ is a map $f\colon N\to M$ such that
$$f(\lambda+n)=\lambda+f(n),$$ for all $\lambda\in L,n\in N$.  
\end{definition}

\begin{remark} If $N$ is any affine space over $L_1\ins L$, then the condition $N\ins M$ is equivalent to the (apparently) weaker hypothesis that $N\cap M\neq\varnothing$.
\end{remark}

\begin{remark}
A morphism of affine spaces is necessarily bijective; therefore, a morphism of affine spaces is the same as an isomorphism of affine spaces.  
\end{remark}

\begin{corollary}
	Identifying Kasteleyn orientations with their associated cochains, the set of equivalence classes $\widetilde{\mathcal{K}}$ is an affine space over $H^1(X)$. 
\end{corollary}

For the remainder of the section we will work with $\mathbf{Z}/2\mathbf{Z}$-valued homology and suppress it from the notation. 
  
\begin{definition}
	A \emph{dimer configuation} on $\Gamma\ins X$ is a collection of edges $D$ such that each vertex of $\Gamma$ is adjacent to exactly one edge in $D$. If $D$ and $D'$ are two dimer configurations, the connected components of the symmetric difference $$(D\cup D')-(D\cap D')$$ are called the \emph{$(D,D')$-composition cycles}. Let $\Delta(D,D')$ denote the corresponding class in $H_1(X)$; we say that $D$ and $D'$ are \emph{equivalent dimer configurations} if $\Delta(D,D')=0$. 
	\end{definition}

\begin{definition}
	A quadratic form on $H_1(X)$ is a function $\omega\colon H_1(X)\to\mathbf{Z}/2\mathbf{Z}$ such that
$$\omega(a+b)=\omega(a)+\omega(b)+a\cdot b,$$
where the $\cdot$ denotes the usual intersection product. Denote by $\mathcal{Q}$ the set of all quadratic forms on $H_1(X)$. 
\end{definition}

Using the universal coefficient theorem, there is a natural identification $H^1(X)\cong (H_1(X))^*$. Therefore, given an element $\delta\in H^1(X)$ and $\omega\in\mathcal{Q}$, we obtain a new function on $H_1(X)$ given by $\delta+\omega$. Since $\delta$ is a homomorphism and $\omega$ is quadratic, we find for $a,b\in H_1(X)$ that
\begin{eqnarray}
(\delta+\omega)(a+b)&=&\delta(a)+\delta(b)+\omega(a)+\omega(b)+a\cdot b\nonumber\\
&=&(\delta+\omega)(a)+(\delta+\omega)(b)+a\cdot b,\nonumber
\end{eqnarray}
from which it follows that $\delta+\omega$ is another quadratic form. With a little more work, it follows that $\mathcal{Q}$ is another affine space over $H^1(X)$. 

Cimasoni and Reshetikhin show that the affine spaces $\widetilde{K}$ and $\mathcal{Q}$ are isomorphic to each other in a natural manner. Before stating their results, we introduce some notation. Fix a Kasteleyn orientation for $\Gamma\ins X$, as well as a dimer configuration $D$. Let $C$ denote an oriented closed curve in $\Gamma$. If $e$ is an edge of $\Gamma$, we define $\epsilon^C_K(e)$ to be $0$ if the orientation of $e$ given by $C$ agrees with that of $K$, and $1$ otherwise. Further, define $\epsilon_K(C)$ as follows:  $$\epsilon_K(C)=\sum_{e\in C}\epsilon^C_K(e).$$
	Finally, let $\ell_D(C)$ be the number of vertices in $C$ whose adjacent dimer of $D$ sticks out to the left of $C$ in $X$. The following results are proved in \cite{Cimasoni:2007}.
	
\begin{theorem}
	Given $\alpha\in H_1(X)$, represent it by oriented simple closed curves $C_1,\dots, C_m$ in $\Gamma$. If $K$ is a Kasteleyn orientation on $\Gamma$, then the function $$q_D^K\colon H_1(X)\to\mathbf{Z}/2\mathbf{Z}$$ defined by 
	
	\begin{equation}\label{quadraticform}
	q_D^K(\alpha)=\sum_{i<j}C_i\cdot C_j+\epsilon_K(C_i)+\ell_D(C_i)+m
		\end{equation}

	is a well-defined quadratic form on $H_1(X)$.
\end{theorem}

\begin{proposition}
\begin{enumerate}
	\item Let $D$ be a fixed dimer configuration on $\Gamma$. If $K$ and $K'$ are two Kasteleyn orientations on $\Gamma$, then $q_D^K-q_D^{K'}$ maps to $\epsilon_K-\epsilon_{K'}$ via the canonical isomorphism $$\textrm{Hom}\big(H_1(X),\mathbf{Z}/2\mathbf{Z}\big)\cong H^1(X).$$
	\item Let $K$ be a fixed Kasteleyn orientation on $\Gamma$. If $D$ and $D'$ are two dimer configurations, then $q_D^K-q_{D'}^K\in\textrm{Hom}\big(H_1(X),\mathbf{Z}/2\mathbf{Z}\big)$ is given by $\alpha\mapsto \alpha\cdot\Delta(D,D')$. 
\end{enumerate}	
\end{proposition}

\begin{corollary}
	Any dimer configuration $D$ on a surface graph $\Gamma\ins X$ induces an isomorphism of affine $H^1(X)$-spaces
	\begin{eqnarray*}
		\psi_D\colon\widetilde{\mathcal{K}}&\to& \mathcal{Q}\\
		\left[\epsilon_K\right] &\mapsto & q_D^K
	\end{eqnarray*}
 from the set of equivalence classes of Kasteleyn orientation on $\Gamma$ onto the set of quadratic forms on the $H_1(X)$. Furthermore, $\psi_D=\psi_{D'}$ if and only if $D$ and $D'$ are equivalent dimer configurations. 
\end{corollary}

\begin{remark}
	It is not hard to show that any two affine spaces over an Abelian group are isomorphic. The significance of the above result is not that the the two spaces are isomorphic, but rather that the specific map $\psi_D$ gives such an isomorphism. \end{remark}

Next, we review the basics of spin structures as described by Johnson in \cite{Johnson:1980}. 
Denote by $UX$  the unit tangent bundle of $X$, and consider the homology/cohomology groups $H_1(UX)$ and $H^1(UX)$. An element of $H_1(UX)$ can be represented as a smooth closed curve in $UX$, which is equivalent to a \emph{framed} closed curve in $X$, that is, a smooth closed curve in $X$, together with a smooth vector field along it. Let $\iota\colon S^1\to UX$ and $p\colon UX\to X$ be the fibre inclusion and projection. These maps yield the following short exact sequences in homology/cohomology:
$$\xymatrix{0\ar[r]&\mathbf{Z}/2\mathbf{Z}\ar^{\iota^*}[r]&H_1(UX)\ar^{p^*}[r]&H_1(X)\ar[r]&0\\
0\ar[r]&H^1(X)\ar^{p^*}[r]&H^1(UX)\ar^{\iota^*}[r]&\mathbf{Z}/2\mathbf{Z}\ar[r]&0}.	$$

The generator of $\mathbf{Z}/2\mathbf{Z}$ in the first sequence will be denoted by $z$; it is the ``fibre class'', and may be represented by a small circle in $X$ with tangential framing. The generator of $\mathbf{Z}/2\mathbf{Z}$ in the second sequence is denoted by $1$. This sequence allows us to identify $H^1(X)$ with the kernel of $\iota^*$. 

\begin{definition}
	A \emph{spin structure} on $X$ is a class $\xi\in H^1(UX)$ such that $\iota^*(\xi)=1$, i.e, $\xi(z)=1$. The set of all spin structures on $X$ will be denoted by $\mathcal{S}(X)$.  
\end{definition}

As Johnson explains, one should  think of such a $\xi$ as a function that assigns a number mod $2$ to each framed curve of $X$, subject to the usual homological conditions, in such a way that the boundary of a disc in $X$, tangentially framed, receives a value of $1$. Note that $H^1(UX)$ is the disjoint union of $H^1(X)\ins H^1(UX)$ and $\mathcal{S}(X)$. 
The following proposition is a summary of the results needed from \cite{Johnson:1980}. 
\begin{proposition}
	If $\alpha\in H_1(X)$ is represented as the sum of $\alpha_i$, each a simple closed curve, and if $\vec{\alpha_i}$ denotes the curve $\alpha_i$ with tangential framing, then the element
$$\tilde{\alpha}=\sum_{i=1}^m \vec{\alpha_i}+mz\in H^1(UX,\mathbf{Z}/2\mathbf{Z}),$$ is well-defined and called the \emph{canonical lift} of $a$.
If $\xi\in\mathcal{S}(X)$, then the function 
\begin{eqnarray}
	q_\xi\colon H_1(X)&\to &\mathbf{Z}/2\mathbf{Z}\nonumber\\
	q_\xi(c)&\mapsto &\xi(\tilde{c})\nonumber
\end{eqnarray}
 is a well-defined quadratic form on $H_1(X)$. The map $\xi\mapsto q_\xi$ is an isomorphism of affine $H^1(X)$-spaces. 
\end{proposition}

Cimasoni and Reshetikhin  combine this isomorphism with their earlier isomorphism to produce an isomorphism between the space of equivalence classes of Kasteleyn orientations and spin structures. The following proposition summarizes the results that we need from \cite{Cimasoni:2007}.

\begin{proposition}
Given a Kasteleyn orienation $K$ and dimer configuration $D$, there is a spin structure $\xi_{K,D}$ that satisfies $q_{\xi_{K,D}}=q_D^K$. Therefore, any choice of dimer configuration $D$ induces an isomorphism of affine $H^1(X)$-spaces
$$\Psi_D\colon\widetilde{K}\to\mathcal{S}(X).$$
Furthermore, $\Psi_D=\Psi_{D'}$ if and only if $D$ and $D'$ are equivalent dimer configurations.  
\end{proposition}

\subsection{Spin Structures from Odd-dashings}
Using the cubical cohomology machinery, a classification of the odd dashings on an Adinkra was obtained in \cite{Doran2012b}. We review this classification below, and then explain how an odd dashing naturally gives rise to a Kasteleyn orientation. Once we have a Kasteleyn orientation and dimer configuration, the work reviewed above yields a spin structure. 

\begin{definition} A \emph{dashing} $M$ of an Adinkra $A$ is a choice, for each edge, of solid or dashed.
	A dashing of $A$ is \emph{odd} if it satisfies the following parity condition: there is an odd number of dashed edges for each $2$-colored loop in $A$. 
	
	The vertex switch at $v$ of a dashing $M$ is the dashing $M'$ obtained from $M$ by switching all of the dashings of the edges incident to $v$. Two odd dashings are considered equivalent if one can be obtained from the other by a sequence of vertex switches. We denote the set of odd dashing on $A$ by $\mathcal{M}$, and the set of equivalence classes by $\widetilde{\mathcal{M}}$. 
\end{definition}


Let $\omega_2\in C^2(A)$ be the $2$-cochain that sends each $2$-cell to $1$, extended linearly.  A dashing on $A$ can be thought of  as the $1$-cochain $\mu\in C^1(A)$ sending each dashed edge to $1$, and the others to $0$. The parity condition on the dashing implies that $d\mu=\omega_2$, and conversely:  a cochain $\mu$ such that $d\mu=\omega_2$ defines an odd dashing. 

\begin{remark}
	The scope of \cite{Doran2012b} was more broad than Adinkras and studied general quotients of $I^N$ by even codes. It is shown that the existence of an odd dashing is equivalent to the vanishing of $w_2=[\omega_2]\in H^2(A,\mathbf{Z}/2\mathbf{Z})$ \cite[Theorem 3.2]{Doran2012b}. This is analogous to the study of spin structures on oriented Riemannian manifolds, where the obstruction to their existence is the second Stiefel-Whitney class in $\mathbf{Z}/2\mathbf{Z}$-valued cohomology. This analogy was one of the motivations for us to attach the structure of complex manifold to an Adinkra in the first place. Moreover, if the second Stiefel-Whitney class vanishes, then the set of spin structures is an affine space over the first $\mathbf{Z}/2\mathbf{Z}$-valued cohomology group.  In what follows, we will see quite clearly that odd dashings give rise to specific spin structures, making this much more than an analogy.
\end{remark}

The next two results from \cite{Doran2012b} completely characterize the set of odd dashings in terms of the cubical complex. 

\begin{proposition} Let $A$ be a cubical complex, $\mu$ an odd dashing, and $T$ a set of vertices. Let $\zeta_T$ be the $0$-cochain that is $1$ on the elements of $T$ and $0$ otherwise. Then the dashing that results from $\mu$ by performing vertex switches at the vertices in $T$ is $\mu+d\zeta_T$. 	
\end{proposition}
 
\begin{theorem}\label{DashTheorem}
	Let $A$ be a cubical complex. Two odd dashings are related by a series of vertex switches if and only if $\mu_2-\mu_1=d\zeta$ for some $\zeta$. Therefore, the set of odd dashings, modulo vertex switches, is in one-to-one correspondence with $H^1(A)$. 
\end{theorem}

These results show that $\widetilde{\mathcal{M}}$ is an affine space over $H^1(A)$. In particular, the number of inequivalent odd dashings on an Adinkra is $2^k$ where $k$ is the dimension of the doubly even code $C$.

To each odd dashing $\mu\in \widetilde{\mathcal{M}}$, we associate a Kasteleyn orientation $K_\mu$ by first giving $A\ins X_\mathcal{A}$ the bipartite orientation, and then reversing the orientation of the dashed edges. The parity condition for the odd dashing implies the parity condition for the resulting orientation to be Kasteleyn. On the level of co-chains, this association is the natural inclusion $C^1(A)\to C^1(X)$ of co-chain complexes. We write $\epsilon_\mu$ for $\epsilon_{K_\mu}$.

\begin{lemma}
	The inclusion $i\colon X_\mathcal{A}\to C(A)$ induces an isomorphism in cohomology in degree one, and an inclusion $i^*\colon H^1(A)\to H^1(X)$.
\end{lemma}
\begin{proof}
This follows by studying the inclusion of chain complexes studied earlier, and applying the left exact functor $\Hom(-,\mathbf{Z}/2\mathbf{Z})$.	
\end{proof}
We identify $H^1(A)$ as a subgroup of $H^1(X)$ via this identification. 

\begin{proposition}
The map $\mu\mapsto \epsilon_\mu$ is well-defined on equivalence classes and induces an inclusion of $\widetilde{\mathcal{M}}$ into $\widetilde{K}$ such that $\widetilde{\mathcal{M}}$, as an affine space over $H^1(A)$, is an affine subspace of $\widetilde{K}$ over $H^1(X)$. 
\end{proposition}

\begin{proof}
	We have already seen that an odd dashing $\mu$ gives rise to a Kasteleyn orientation. This is seen alternatively by noting that the the $2$-cocycle $\omega_2\in H^2(A)$ maps to the cocycle $\omega_2^X\in H^2(X)$, where the superscript has been added to avoid confusion. Since $d\mu=\omega_2$, we have $d(i^*\mu)=\omega_2^X$, so that $i^*(\mu)$ is a Kasteleyn orientation. Equivalent odd dashings are mapped to equivalent Kasteleyn orientations because of the identification $C^0(A)=C^0(X)$. Moreover, the identification also shows that if $\epsilon_{\mu}$ and $\epsilon_{\mu'}$ are equivalent Kasteleyn orientations, then $\mu$ and $\mu'$ are equivalent odd dashings. Therefore, $i^*$ induces the desired inclusion $\widetilde{\mathcal{M}}\to\widetilde{K}$.  
	
	Finally, let $\nu\in H^1(A)$ and let $\mu$ be an odd dashing. Since $i^*$ is a homomorphism, and $H^1(A)$ is identified with its image under $i^*$, it follows that the addition map is compatible, and that $\widetilde{\mathcal{M}}$ is an affine subspace of $\widetilde{K}$. 
\end{proof}

In order to obtain spin structures on $X_\mathcal{A}$ via the work of Cimasoni and Reshetikhin, we need to choose a dimer configuration. There are $N$ natural choices for dimer configurations on $A$, namely the sets $D_j$ consisting of the edges of color $j$. As it turns out, these $N$ choices of dimer configurations are equivalent to each other. 

\begin{proposition}
The dimer configurations $D_j$ are all equivalent to each other. 	
\end{proposition}
\begin{proof}
Consider $\Delta(D_j,D_{j'})$ where $j'$ is the next color in the rainbow defining $X_\mathcal{A}$. Since $D_j$ and $D_{j'}$ have no edges in common, $\Delta(D_j,D_{j'})$ is represented by the union of the boundaries of all $j/j'$ colored faces in $X_\mathcal{A}$. By design, these are all contractible, from which it follows $\Delta(D_j,D_{j'})=0$, i.e., $D_j$ and $D_{j'}$ are equivalent. This holds for all pairs of adjacent colors; we get the full result using the transitivity of the equivalence relation.  	
\end{proof}

 While there are lots of dimer configurations on $X_\mathcal{A}$, we see that there is one natural choice of dimer configuration on our Adinkra up to equivalence. Therefore, the set of equivalence classes of odd dashings $\widetilde{\mathcal{M}}$ gives rise naturally to a subset of spin structures on the Riemann surface $X_\mathcal{A}$.
More precisely, we have the following corollary:

\begin{corollary}
The map $\psi\circ i^*\colon \widetilde{\mathcal{M}}\to\mathcal{S}(X_\mathcal{A})$ is an inclusion of affine spaces, where $\psi$ is any choice of $\psi_{D_j}$. In particular, the set of equivalence classes of odd dashings on $A$ can be identified with a subset of $\mathcal{S}(X_\mathcal{A})$. 
\end{corollary}

\begin{remark}
	The specific subset of spin structures obtained above is an affine space over the doubly even code $C$ associated to $A$ after identifying $H_1(A)$ with $C$ as described in \cite{Doran2012b}.
\end{remark}

\begin{example}
As an example, we determine the set of spin structures coming from odd dashings on the Riemann surface $X_4$ associated to the $4$-cube for the rainbow $(1,2,3,4)$. 
If we fix a white vertex $w$, then the $1/3$ color loop and $2/4$ color loop are the standard $a,b$ cycles that define a homology basis for $H_1(X)$. A basis for $H_1(UX)$ is given by  $\vec{a},\vec{b}$, and $z$. The set of all spin structures on $X_4$ is given by the set of $\xi\in H^1(UX)$ such that $\xi(z)=1$. This yields four spin structures which are completely determined by the images of $\vec{a}$ and $\vec{b}$. 

Up to equivalence, there is a unique odd dashing on the hypercube. It is depicted in Figure \ref{HyperCubeExample}. Our choice of dimer does not matter at this point, so we do not explicitly fix it. According to Equation \eqref{quadraticform}, which defines $q^K_D$, all $2$-colored loops must map to $0$. It follows that the quadratic form on $H_1(X)$ associated to the unique odd dashing is determined by setting $q(a)=q(b)=0$. 

On the other hand, for each of the spin structures on $X_4$, not necessarily arising from our construction, we can compute the associated quadratic form. The spin structure that yields the same quadratic form as above is the spin structure $\xi$ defined by sending both $\vec{a}$ and $\vec{b}$ to $1$. Notice that the Arf-invariant for this form is $0$, so that we have obtained one of the three even spin structures. 

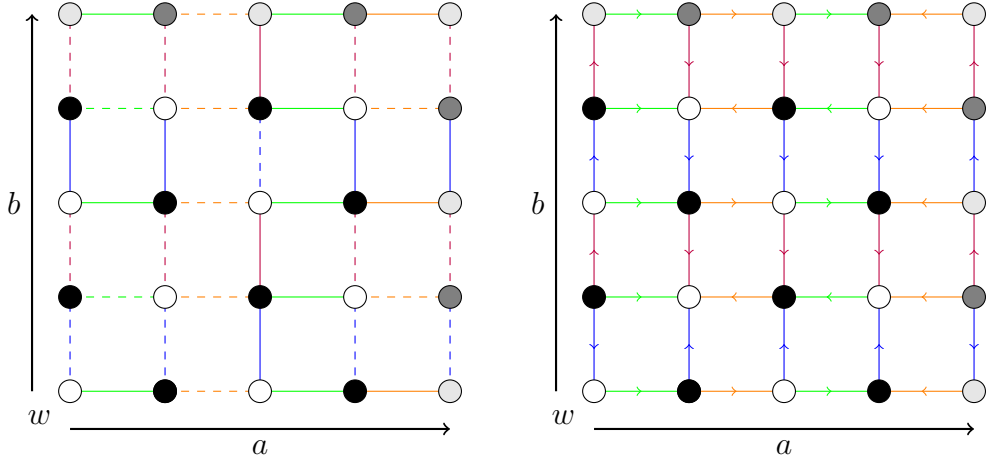
\begin{figure}
\centering
\subfloat{
\begin{tikzpicture}[scale=2.5]
	\draw[dashed,color=blue] (-1,-1)--(-1,-0.5);
	\draw[color=blue] (-1,0)--(-1,0.5);
	\draw[dashed,color=blue] (-0.5,-1)--(-0.5,-0.5);
	\draw[color=blue] (-0.5,0)--(-0.5,0.5);
	\draw[color=blue] (-0,-1)--(0,-0.5);
	\draw[dashed,color=blue] (0,0)--(0,0.5);
	\draw[dashed, color=blue] (0.5,-1)--(0.5,-0.5);
	\draw[color=blue] (0.5,0)--(0.5,0.5);
	\draw[dashed, color=blue] (1,-1)--(1,-0.5);
	\draw[color=blue] (1,0)--(1,0.5);
	
	\draw[dashed, color=purple] (-1,-0.5)--(-1,0);
	\draw[dashed, color=purple] (-1,0.5)--(-1,1);
	\draw[dashed, color=purple] (-0.5,-0.5)--(-0.5,0);
	\draw[dashed, color=purple] (-0.5,0.5)--(-0.5,1);
	\draw[color=purple] (0,-0.5)--(0,0);
	\draw[color=purple] (0,0.5)--(0,1);
	\draw[dashed, color=purple] (0.5,-0.5)--(0.5,0);
	\draw[dashed, color=purple] (0.5,0.5)--(0.5,1);
	\draw[dashed, color=purple] (1,-0.5)--(1,0);
	\draw[dashed, color=purple] (1,0.5)--(1,1);
	
	\draw[color=green] (-1,-1)--(-0.5,-1);
	\draw[dashed,color=green] (-1,-0.5)--(-0.5,-0.5);
	\draw[color=green] (-1,0)--(-0.5,0);
	\draw[dashed,color=green] (-1,0.5)--(-0.5,0.5);
	\draw[color=green] (-1,1)--(-0.5,1);
	\draw[color=green] (0,-1)--(0.5,-1);
	\draw[color=green] (0,-0.5)--(0.5,-0.5);
	\draw[color=green] (0,0)--(0.5,0);
	\draw[color=green] (0,0.5)--(0.5,0.5);
	\draw[color=green] (0,1)--(0.5,1);

	\draw[dashed, color=orange] (-0.5,-1)--(0,-1);
	\draw[dashed, color=orange] (-0.5,-0.5)--(0,-0.5);
	\draw[dashed, color=orange] (-0.5,0)--(0,0);
	\draw[dashed, color=orange] (-0.5,0.5)--(0,0.5);
	\draw[dashed, color=orange] (-0.5,1)--(0,1);
	\draw[color=orange] (0.5,-1)--(1,-1);
	\draw[dashed, color=orange] (0.5,-0.5)--(1,-0.5);
	\draw[color=orange] (0.5,0)--(1,0);
	\draw[dashed, color=orange] (0.5,0.5)--(1,0.5);
	\draw[color=orange] (0.5,1)--(1,1);
	
	\draw[fill=white] (-1,-1) node[below left=3pt] {$w$} circle[radius=0.06];
	\draw[fill=white] (-1,0) circle[radius=0.06];
	\draw[fill=gray!20] (-1,1) circle[radius=0.06];
	\draw[fill=white] (0,-1) circle[radius=0.06];
	\draw[fill=white] (0,0) circle[radius=0.06];
	\draw[fill=gray!20] (0,1) circle[radius=0.06];
	\draw[fill=gray!20] (1,-1) circle[radius=0.06];
	\draw[fill=gray!20] (1,0) circle[radius=0.06];
	\draw[fill=gray!20] (1,1) circle[radius=0.06];
	\draw[fill=white] (-0.5,-0.5) circle[radius=0.06];
	\draw[fill=white] (-0.5,0.5) circle[radius=0.06];
	\draw[fill=white] (0.5,0.5) circle[radius=0.06];
	\draw[fill=white] (0.5,-0.5) circle[radius=0.06];

	\draw[fill=black] (-0.5,-1) circle[radius=0.06];
	\draw[fill=black] (-1,0.5) circle[radius=0.06];
	\draw[fill=black] (-1,-0.5) circle[radius=0.06];
	\draw[fill=black] (-0.5,-1) circle[radius=0.06];
	\draw[fill=black] (-0.5,0) circle[radius=0.06];
	\draw[fill=gray] (-0.5,1) circle[radius=0.06];
	\draw[fill=black] (0,-0.5) circle[radius=0.06];
	\draw[fill=black] (0,0.5) circle[radius=0.06];
	\draw[fill=black] (0.5,-1) circle[radius=0.06];
	\draw[fill=black] (0.5,0) circle[radius=0.06];
	\draw[fill=gray] (0.5,1) circle[radius=0.06];
	\draw[fill=gray] (1,0.5) circle[radius=0.06];
	\draw[fill=gray] (1,-0.5) circle[radius=0.06];
	
	\draw[->, thick] (-1,-1.2)--node[anchor=north] {$a$} (1,-1.2);
	\draw[->, thick] (-1.2,-1)--node[anchor=east] {$b$} (-1.2,1);
\end{tikzpicture}}	
\hfil
\subfloat{\begin{tikzpicture}[scale=2.5]
	\draw[middlearrow={<},color=blue] (-1,-1)--(-1,-0.5);
	\draw[middlearrow={>},color=blue] (-1,0)--(-1,0.5);
	\draw[middlearrow={>},color=blue] (-0.5,-1)--(-0.5,-0.5);
	\draw[middlearrow={<},color=blue] (-0.5,0)--(-0.5,0.5);
	\draw[middlearrow={>},color=blue] (0,-1)--(0,-0.5);
	\draw[middlearrow={<},color=blue] (0,0)--(0,0.5);
	\draw[middlearrow={>},color=blue] (0.5,-1)--(0.5,-0.5);
	\draw[middlearrow={<},color=blue] (0.5,0)--(0.5,0.5);
	\draw[middlearrow={<},color=blue] (1,-1)--(1,-0.5);
	\draw[middlearrow={>},color=blue] (1,0)--(1,0.5);
	
	\draw[middlearrow={>},color=purple] (-1,-0.5)--(-1,0);
	\draw[middlearrow={>},color=purple] (-1,0.5)--(-1,1);
	\draw[middlearrow={<}, color=purple] (-0.5,-0.5)--(-0.5,0);
	\draw[middlearrow={<}, color=purple] (-0.5,0.5)--(-0.5,1);
	\draw[middlearrow={<}, color=purple] (0,-0.5)--(0,0);
	\draw[middlearrow={<}, color=purple] (0,0.5)--(0,1);
	\draw[middlearrow={<}, color=purple] (0.5,-0.5)--(0.5,0);
	\draw[middlearrow={<}, color=purple] (0.5,0.5)--(0.5,1);
	\draw[middlearrow={>}, color=purple] (1,-0.5)--(1,0);
	\draw[middlearrow={>}, color=purple] (1,0.5)--(1,1);
	
	\draw[middlearrow={>},color=green] (-1,-1)--(-0.5,-1);
	\draw[middlearrow={>},color=green] (-1,-0.5)--(-0.5,-0.5);
	\draw[middlearrow={>},color=green] (-1,0)--(-0.5,0);
	\draw[middlearrow={>},color=green] (-1,0.5)--(-0.5,0.5);
	\draw[middlearrow={>},color=green] (-1,1)--(-0.5,1);
	\draw[middlearrow={>},color=green] (0,-1)--(0.5,-1);
	\draw[middlearrow={<},color=green] (0,-0.5)--(0.5,-0.5);
	\draw[middlearrow={>},color=green] (0,0)--(0.5,0);
	\draw[middlearrow={<},color=green] (0,0.5)--(0.5,0.5);
	\draw[middlearrow={>}, color=green] (0,1)--(0.5,1);

	\draw[middlearrow={>},color=orange] (-0.5,-1)--(0,-1);
	\draw[middlearrow={<},color=orange] (-0.5,-0.5)--(0,-0.5);
	\draw[middlearrow={>},color=orange] (-0.5,0)--(0,0);
	\draw[middlearrow={<},color=orange] (-0.5,0.5)--(0,0.5);
	\draw[middlearrow={>},color=orange] (-0.5,1)--(0,1);
	\draw[middlearrow={<},color=orange] (0.5,-1)--(1,-1);
	\draw[middlearrow={<},color=orange] (0.5,-0.5)--(1,-0.5);
	\draw[middlearrow={<},color=orange] (0.5,0)--(1,0);
	\draw[middlearrow={<},color=orange] (0.5,0.5)--(1,0.5);
	\draw[middlearrow={<},color=orange] (0.5,1)--(1,1);
	
	\draw[fill=white] (-1,-1) node[below left=3pt] {$w$} circle[radius=0.06];
	\draw[fill=white] (-1,0) circle[radius=0.06];
	\draw[fill=gray!20] (-1,1) circle[radius=0.06];
	\draw[fill=white] (0,-1) circle[radius=0.06];
	\draw[fill=white] (0,0) circle[radius=0.06];
	\draw[fill=gray!20] (0,1) circle[radius=0.06];
	\draw[fill=gray!20] (1,-1) circle[radius=0.06];
	\draw[fill=gray!20] (1,0) circle[radius=0.06];
	\draw[fill=gray!20] (1,1) circle[radius=0.06];
	\draw[fill=white] (-0.5,-0.5) circle[radius=0.06];
	\draw[fill=white] (-0.5,0.5) circle[radius=0.06];
	\draw[fill=white] (0.5,0.5) circle[radius=0.06];
	\draw[fill=white] (0.5,-0.5) circle[radius=0.06];

	\draw[fill=black] (-0.5,-1) circle[radius=0.06];
	\draw[fill=black] (-1,0.5) circle[radius=0.06];
	\draw[fill=black] (-1,-0.5) circle[radius=0.06];
	\draw[fill=black] (-0.5,-1) circle[radius=0.06];
	\draw[fill=black] (-0.5,0) circle[radius=0.06];
	\draw[fill=gray] (-0.5,1) circle[radius=0.06];
	\draw[fill=black] (0,-0.5) circle[radius=0.06];
	\draw[fill=black] (0,0.5) circle[radius=0.06];
	\draw[fill=black] (0.5,-1) circle[radius=0.06];
	\draw[fill=black] (0.5,0) circle[radius=0.06];
	\draw[fill=gray] (0.5,1) circle[radius=0.06];
	\draw[fill=gray] (1,0.5) circle[radius=0.06];
	\draw[fill=gray] (1,-0.5) circle[radius=0.06];
	
	\draw[->, thick] (-1,-1.2)--node[anchor=north] {$a$} (1,-1.2);
	\draw[->, thick] (-1.2,-1)--node[anchor=east] {$b$} (-1.2,1);
	
\end{tikzpicture}}
\caption{A fundamental domain for the Riemann surface $X_4$ with its unique odd dashing on the left, and the corresponding Kasteleyn orientation on the right. The rainbow depicted here is (blue,orange,purple,green).}
\label{HyperCubeExample}
\end{figure}
\end{example}

\begin{example} As a more complicated example, we determine the set of spin structures coming from odd dashings on the Riemann surface $X_\mathcal{A}$, where $\mathcal{A}$ is the Adinkra corresponding to the unique non-trivial doubly even code $C=\gen{(1,1,1,1)}$, with the rainbow $(1,2,3,4)$. Since $C$ is $1$-dimensional, there are exactly $2$ odd dashings up to equivalence. Fixing a white vertex $w$, the homology group $H_1(A)$ is identified with the doubly even code $C$, as described earlier; the group $H^1(A)$ is dual to this group. 
	
We realize $X_\mathcal{A}$ as a square torus. In fact, we can take a subset of the fundamental domain for $X_4$ as a fundamental domain for $X_\mathcal{A}$. Such a fundamental domain is depicted in Figure \ref{QuotientExample}.
 Fixing a white vertex $w$, the cycles obtained by traversing colors $1,2,3,4$ in sequence, and in reverse, are homologous to the $a$ and $b$ cycles (see Figure \ref{QuotientExample}). Note, these cycles have nothing to do with the $a$ and $b$ cycles above, despite the notation being used. A basis for $H_1(UX)$ is given by $\vec{a},\vec{b},z$, and a spin structure $\xi$ is an assignment of $\vec{a},\vec{b}$ to $\mathbf{Z}/2\mathbf{Z}$. With the odd dashing $\mu$ that is depicted in Figure \ref{QuotientExample}, we obtain a spin-structure $\xi_\mu$. Now we compute the quadratic form $q_\mu$ associated to $\mu$, with the dimer $D_1$ using Equation \eqref{quadraticform}. It suffices to find the values of $q_\mu(a)$ and $q_\mu(b)$. Using the fact that $a$ is homologous to $P^c$, where $c=(1,1,1,1)$, we can read off from Figure \ref{QuotientExample} that
 $$\epsilon_{K_\mu}(a)=1,\ \ \textrm{and}\ \ \ell_D(a)=1.$$
 It follows that $q_\mu(a)=1$. On the other hand, the cycle $b$ is homologous to the path starting at $w$ and travelling the edges in the order $(1,4,3,2)$. From Figure \ref{QuotientExample}, we read off that
 $$\epsilon_{K_\mu}(b)=0,\ \ \textrm{and}\ \ \ell_D(b)=1,$$ from which it follows that $q_\mu(b)=0$. 
 
 Therefore, this spin structure corresponds to the spin structure that sends $\vec{a}$ to $0$ and $\vec{b}$ to $1$. We obtain the second spin structure by adding the non-trivial element $\eta$ of $H^1(A)$ to $q_\mu$. Since $H^1(A)$ is dual to $H_1(A)$, $\eta$ is the element defined by sending $P^c$ to $1$. Both the $a$ and $b$ cycles on $X_\mathcal{A}$ are homologous, in the cubical complex, to $P^c$, from which it follows that $\eta(a)=\eta(b)=1$. Therefore, the second spin structure obtained is the one defined by sending $\vec{a}$ to $1$ and $\vec{b}$ to $0$. The Arf-invariants of both spin structures are $0$, so that we have obtained two of the three even spin structures on the torus. 
 
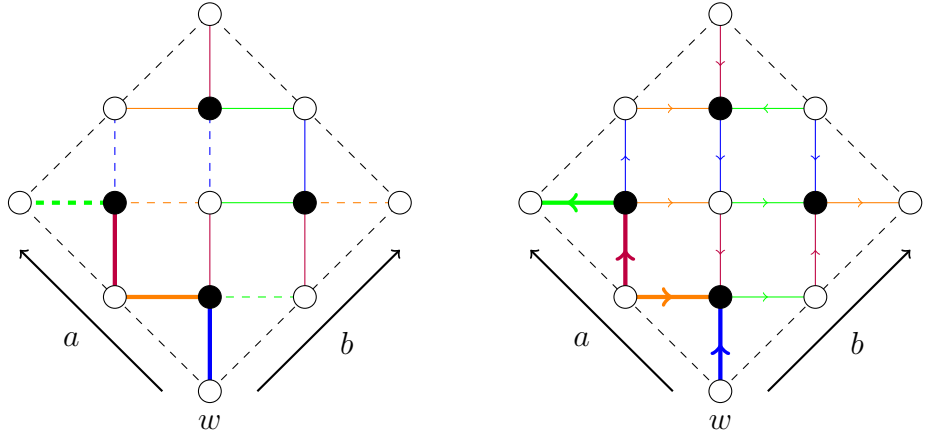
\begin{figure}
\centering
\subfloat{
\begin{tikzpicture}[scale=2.5]

	\draw[dashed, color=blue] (-0.5,0)--(-0.5,0.5);
	\draw[color=blue, ultra thick] (0,-1)--(0,-0.5);
	\draw[dashed, color=blue] (0,0)--(0,0.5);
	\draw[color=blue] (0.5,0)--(0.5,0.5);
	
	\draw[color=purple, ultra thick] (-0.5,-0.5)--(-0.5,0);
	\draw[color=purple] (0,-0.5)--(0,0);
	\draw[color=purple] (0,0.5)--(0,1);
	\draw[color=purple] (0.5,-0.5)--(0.5,0);
	
	\draw[dashed,color=green, ultra thick] (-1,0)--(-0.5,0);
	\draw[dashed, color=green] (0,-0.5)--(0.5,-0.5);
	\draw[color=green] (0,0)--(0.5,0);
	\draw[color=green] (0,0.5)--(0.5,0.5);

	\draw[color=orange, ultra thick] (-0.5,-0.5)--(0,-0.5);
	\draw[dashed,color=orange] (-0.5,0)--(0,0);
	\draw[color=orange] (-0.5,0.5)--(0,0.5);
	\draw[dashed, color=orange] (0.5,0)--(1,0);

	\draw[dashed] (0,-1)--(1,0);
	\draw[dashed] (0,-1)--(-1,0);
	\draw[dashed] (1,0)--(0,1);
	\draw[dashed] (-1,0)--(0,1);
	
	\draw[->,thick] (0.25,-1)--node[anchor=north west] {$b$}(1,-0.25);
	\draw[->, thick] (-0.25,-1)--node[anchor=north east] {$a$}(-1,-0.25);

	\draw[fill=white] (-1,0) circle[radius=0.06];
	\draw[fill=white] (0,-1) node[below=5pt] {$w$} circle[radius=0.06];
	\draw[fill=white] (0,0) circle[radius=0.06];
	\draw[fill=white] (0,1) circle[radius=0.06];
	\draw[fill=white] (1,0) circle[radius=0.06];
	\draw[fill=white] (-0.5,-0.5) circle[radius=0.06];
	\draw[fill=white] (-0.5,0.5) circle[radius=0.06];
	\draw[fill=white] (0.5,0.5) circle[radius=0.06];
	\draw[fill=white] (0.5,-0.5) circle[radius=0.06];

	\draw[fill=black] (-0.5,0) circle[radius=0.06];
	\draw[fill=black] (0,-0.5) circle[radius=0.06];
	\draw[fill=black] (0,0.5) circle[radius=0.06];
	\draw[fill=black] (0.5,0) circle[radius=0.06];

\end{tikzpicture}}
\hfil
\subfloat{
\begin{tikzpicture}[scale=2.5]

	\draw[middlearrow={>},color=blue] (-0.5,0)--(-0.5,0.5);
	\draw[middlearrow={>},color=blue, ultra thick] (0,-1)--(0,-0.5);
	\draw[middlearrow={<},color=blue] (0,0)--(0,0.5);
	\draw[middlearrow={<},color=blue] (0.5,0)--(0.5,0.5);
	
	\draw[middlearrow={>},color=purple, ultra thick] (-0.5,-0.5)--(-0.5,0);
	\draw[middlearrow={<},color=purple] (0,-0.5)--(0,0);
	\draw[middlearrow={<},color=purple] (0,0.5)--(0,1);
	\draw[middlearrow={>},color=purple] (0.5,-0.5)--(0.5,0);
	
	\draw[middlearrow={<},color=green, ultra thick] (-1,0)--(-0.5,0);
	\draw[middlearrow={>},color=green] (0,-0.5)--(0.5,-0.5);
	\draw[middlearrow={>},color=green] (0,0)--(0.5,0);
	\draw[middlearrow={<},color=green] (0,0.5)--(0.5,0.5);

	\draw[middlearrow={>},color=orange, ultra thick] (-0.5,-0.5)--(0,-0.5);
	\draw[middlearrow={>},color=orange] (-0.5,0)--(0,0);
	\draw[middlearrow={>},color=orange] (-0.5,0.5)--(0,0.5);
	\draw[middlearrow={>},color=orange] (0.5,0)--(1,0);

	\draw[dashed] (0,-1)--(1,0);
	\draw[dashed] (0,-1)--(-1,0);
	\draw[dashed] (1,0)--(0,1);
	\draw[dashed] (-1,0)--(0,1);
	
	\draw[->,thick] (0.25,-1)--node[anchor=north west] {$b$}(1,-0.25);
	\draw[->, thick] (-0.25,-1)--node[anchor=north east] {$a$}(-1,-0.25);

	\draw[fill=white] (-1,0) circle[radius=0.06];
	\draw[fill=white] (0,-1) node[below=5pt] {$w$} circle[radius=0.06];
	\draw[fill=white] (0,0) circle[radius=0.06];
	\draw[fill=white] (0,1) circle[radius=0.06];
	\draw[fill=white] (1,0) circle[radius=0.06];
	\draw[fill=white] (-0.5,-0.5) circle[radius=0.06];
	\draw[fill=white] (-0.5,0.5) circle[radius=0.06];
	\draw[fill=white] (0.5,0.5) circle[radius=0.06];
	\draw[fill=white] (0.5,-0.5) circle[radius=0.06];

	\draw[fill=black] (-0.5,0) circle[radius=0.06];
	\draw[fill=black] (0,-0.5) circle[radius=0.06];
	\draw[fill=black] (0,0.5) circle[radius=0.06];
	\draw[fill=black] (0.5,0) circle[radius=0.06];

\end{tikzpicture}}
\caption{A fundamental domain for the surface corresponding to the code $C=\gen{(1111)}$  with a choice of odd dashing on the left, and the corresponding Kasteleyn orientation on the right. The black dashed edges form the boundary, and the $a$ and $b$ cycles determine the gluing. The path in bold, starting at $w$ is the cycle $P^c$, where $c=(1,1,1,1)$. }
\label{QuotientExample}
\end{figure}

\end{example}

\begin{remark}
 The interested reader can verify the above result in a different way by working with the above bipartite graph, but choosing the other odd dashing to compute the second spin structure. Such an odd dashing can be obtained from the one depicted in Figure \ref{QuotientExample} by switching the dashing on the edges of the $a$ and $b$ cycles for all colors other than blue. As long as we agree on the $a$ and $b$ cycles, the results will be same. 

\end{remark}

\begin{remark}
	We emphasize here that the above construction is valid for any choice of rainbow. 
	While choosing two different rainbows may give rise to very different Riemann surfaces (see Section 4), the odd dashings and dimer configurations on an Adinkra are purely graph-theoretic and do not depend on such a choice. 
	Therefore, even though varying the rainbow will change the surface on which the spin structures are supported, we may still think of the set of spin structures coming from odd dashings as being common to \emph{all} of these surfaces. 
\end{remark}

\subsection{Super Riemann Surface Structure}
In this section, we show that the Riemann surfaces $X_\mathcal{A}$ admit natural super Riemann surface structures. We briefly review the required theory below and encourage the reader to consult \cite{Donagi:2013,Witten:2012} for more details.

\begin{definition}
	A $\mathbf{Z}/2\mathbf{Z}$-graded sheaf of algebras $A=A_0\oplus A_1$ is \emph{supercommutative} if it satisfies $fg=(-1)^{ij}gf$ for all $f\in A_i,g\in A_j$. A \emph{supercommutative locally ringed space} is a locally ringed space $(M,\mathcal{O})$ where $\mathcal{O}$ is a a sheaf of supercommutative algebras whose stalks are local rings. 
\end{definition}

\begin{definition}
Let $M$ be  a complex manifold, and  $V$  a vector bundle over $M$. The  supercommutative locally ringed space $S=S(M,V)$ is defined to be the pair $(M,\mathcal{O}_S)$, where $\mathcal{O}_S$ is the sheaf of $\mathcal{O}_M$-valued sections of the exterior algebra $\wedge^\bullet(V^\vee)$ on the dual bundle $V^\vee$.	

A \emph{supermanifold} is a supercommutative locally ringed space, locally isomorphic to some $S(M,V)$. It is called \emph{split} if it is globally isomorphic to such a space. 
\end{definition}
 By virtue of the exterior product, the even part of $\wedge^\bullet(V^\vee)$ is just $\mathcal{O}_M$, while the odd-part is the rest of the direct sum. \begin{example}
	Complex superspace is defined as $$\mathbf{C}^{m|n}=(\mathbf{C}^n,\mathcal{O}_{\mathbf{C}^{n|m}})=S(\mathbf{C}^m,\mathcal{O}_{\mathbf{C}^m}^{\oplus n}).$$ That is, the underlying manifold is usual complex $m$-space, and the vector bundle is the trivial one of rank $n$. 
\end{example}

The sheaf $\mathcal{O}_S$ admits a surjective map to $\mathcal{O}_M$ with kernel $J$ consisting of the ``odd'' functions. It is not always the case that we can give $\mathcal{O}_S$ the structure of $\mathcal{O}_M$-module, but it is the case if $S$ is split. The dual of $V$ can be recovered as $J/J^2$. 

The tangent bundle of $S$ is the sheaf of derivations on $\mathcal{O}_S$. It is a $\mathbf{Z}/2\mathbf{Z}$-graded vector bundle. In the case of $\mathbf{C}^{n|m}$, $TS$ is freely generated by even tangent vectors $\frac{\di}{\di x_i}$ for $i=1,\cdots, m$ and odd tangent vectors $\frac{\di }{\di \theta_j}$ for $j=1,\dots n$. In general, $TS$ does not have such distinguished generators because the even and odd parts fail to be sheaves of $\mathcal{O}_S$-modules. Restricting $TS$ to the underlying manifold $M$, we do have a splitting though: the even part, $TS_+$ is $TM$, the tangent bundle of $M$, and the odd part, $TS_-$, is a vector bundle $V$. In this case, $S$ is modelled on $M,V$, but may not be isomorphic to $S(M,V)$. By definition, the dimension of $S$ is the pair $(m|n)$, where $m,n$ are the ranks of the even/odd pieces of the tangent bundle. 

\begin{definition}
	With notation as above, the \emph{dimension} of the supermanifold $S$ is the pair $(m|n)$, where $m$ is the rank of $TS_+$ and $n$ is the rank of $TS_-$.
\end{definition}

Morphisms of supermanifolds are morphisms of locally ringed spaces that repsect the supergrading. 
Given a supermanifold $S=(M,\mathcal{O}_S)$, and a finite (unbranched) covering map $f\colon\tilde{M}\to M$,  we can construct the supermanifold $\widetilde{S}$ by taking $\widetilde{M}$ as the underlying manifolding, and taking the pull-back of $\mathcal{O}_S$ as sheaf of functions. This yields a morphism $F\colon\tilde{S}\to S$ of supermanifolds whose restriction to the underlying manifold is $f$. 

We can also do the same thing for branched covers. More precisely, given a supermanifold $S=(M,\mathcal{O}_S)$, a branched covering $f$ of $M$ with branch divisor $B\ins M$, and a divisor $D\ins S$ whose intersection with $M$ is $B$, we can construct a supermanifold $\tilde{S}$ with morphism $F$ whose branch divisor is $D$, and whose restriction is $f$.

\begin{definition}
	A \emph{super Riemann surface} is a pair $\mathcal{S}=(S,\mathcal{D})$, where $S=(C,\mathcal{O}_S)$ is complex supermanifold of dimension $(1|1)$ and $\mathcal{D}$ is an everywhere non-integrable odd distribution $\mathcal{D}\ins TS$. 
\end{definition} 

The non-integrablity condition implies that the vector bundle $V$ associated to $S$ is a square-root of the dual of the canonical bundle. That is, the reduced space of a super Riemann surface is a spin curve; conversely, a Riemann surface, together with a choice of spin structure, naturally gives rise to a super Riemann surface.

\begin{proposition}
	A choice of odd dashing on the Adinkra $A$ naturally gives rise to a super Riemann surface structure on $X_\mathcal{A}$ for any choice of rainbow. Equivalent odd dashings give rise to isomorphic super Riemann surfaces. 
\end{proposition}

\begin{proof}
We have already seen that a choice of odd dashing gives rise to a spin structure. The above discussion implies that we obtain a super Riemann surface structure on $X_\mathcal{A}$. 
Equivalent odd dashings give rise to isomorphic line bundles, and such isomorphisms induces isomorphisms of super Riemann surfaces, from which the second statement follows.	
\end{proof}

\begin{proposition}
	Let $A$ be an Adinkra, and let $A'$ be the quotient of $A$ by a doubly even code. Further, choose an odd dashing on $A'$ and lift the dashing to one on $A$. The morphism of Riemann surfaces $X_\mathcal{A}\to X_{\mathcal{A}'}$ induces a morphism of super Riemann surfaces, where $\mathcal{A}$ and $\mathcal{A}'$ are given the same rainbow.  
\end{proposition}

\begin{proof}
	The morphism of Riemann surfaces $X_\mathcal{A}\to X_{\mathcal{A}'}$ induces a homomorphism 
	$$\textrm{Pic}(X_{\mathcal{A}'})\to \textrm{Pic}(X_\mathcal{A})$$ on groups of isomorphism classes of line bundles via pullback. Choosing the odd dashing on $A$ in the manner specified, it follows that the line bundle $\mathcal{L}_{\mathcal{A}'}$ on $X_{\mathcal{A}'}$ maps to the line bundle $\mathcal{L}_\mathcal{A}$ on $X_\mathcal{A}$ via this homomorphism. From this, we can construct a natural map on the structure sheaves, giving rise the required map on the level of super Riemann surfaces. 
\end{proof}

In summary, we have shown that the odd dashing on an Adinkra gives rise to a spin structure on the Riemann surface. This spin structure can be used to enhance our surface into a super Riemann surface, and this construction is compatible with the covering maps arising from the quotients of Adinkras by doubly-even codes. In \cite{Doran:2015}, it was shown that all of the Riemann surfaces associated to Adinkras of a fixed length $N$ and rainbow are branched covers of the Riemann sphere, viewed either as  the Belyi base or the intermediate beachball $B_N$. It is therefore natural to ask whether or not we can give a super Riemann surface structure to these Riemann spheres in such a way that we are able to  lift these structures to the surfaces $X_A$. 

The answer to this question is ``No''. 
While there is a unique choice of spin structure on the Riemann sphere, leading to a super Riemann surface structure, we cannot lift this to a super Riemann surface structure on the surfaces $X_\mathcal{A}$. 
The reason for this is the branching of the corresponding covering maps, which causes issues when trying to lift the non-integrable odd distribution. 
If we instead consider only the structure of a supermanifold, then the choice of supermanifold manifold structure on the base will give rise to supermanifold structures on all of the surface $X_\mathcal{A}$, and all of the morphisms discussed in \cite{Doran:2015} can be lifted to morphisms of supercomplex manifolds. 
These supercomplex manifold structures are, in general, distinct from the ones arising in our construction involving spin structures. 

Donagi and Witten clarify the relationship between super Riemann surface structures on branched covers and the associated Ramond and Neveu-Schwarz punctures depending on the various ramification indices \cite{Witten:2012}. 
It is unclear how to produce the super Riemann surface structures corresponding to odd dashings on Adinkra Riemann surfaces ``from the bottom up'' by considering branched covers of the beachball base.

\section{From Engineering Dimension to Morse Divisor}
In this section, we give a geometric interpretation for the final additional structure on an Adinkra, its height function, via Morse theory (qua topological surface) and as a divisor (qua algebraic curve).
Although Morse functions are traditionally real-valued functions on manifolds, there have been many different attempts to discretize them. 
Two primary approaches to this are that of T. Banchoff \cite{Banchoff:1970} and that of R. Forman \cite{Forman:1998a,Forman:1998b}; a construction of E. Bloch reduces the study of Forman discrete Morse functions to those studied by Banchoff \cite{Bloch:2010}.
In this section we will show that the Adinkra height function defines a discrete Morse function on $X_\mathcal{A}$
in the sense of Banchoff \cite{Banchoff:1970}. The topology of $X_\mathcal{A}$ is captured by the critical behavior of this function, with critical points among the critical points of the Belyi map.  No explicit dependence on complex structure is made in this section.  On the other hand, we proceed to show that the expression for the the Euler characteristic of $X_\mathcal{A}$ in terms of these critical points leads  us naturally to consider a divisor on the Riemann surface, one that we will call a \emph{Morse divisor}. After reviewing some of the basics of Jacobians of algebraic curves, we explain how the collection of all height functions leads naturally to a subset of points on the Jacobian of our surface.  These results are illustrated with a number of examples ($N = 4, 5, 6$ and $8$).

\subsection{From Engineering Dimension to Discrete Morse Function}
As described in \cite{Doran:2015}, the height function corresponds to the engineering dimension of the component fields the vertex represents. Mathematically, we define the height as follows:

\begin{definition}
	If $A$ is an Adinkra and $V(A)$ is its vertex set, then a \emph{height function} on $A$ is a map $h\colon V(A)\to\mathbf{Z}$ such that the heights of adjacent vertices always differ by $1$.
\end{definition} 
We will assume that all of the white vertices have even height, while all of the black vertices have odd height. We will consider two height functions $h_1$ and $h_2$ as being equivalent if they are related by an overall shift. 

Morse theory was developed to understand the topology of a manifold by way of studying the critical behavior of a smooth function on it. Banchoff's approach to discretizing Morse theory assigns an index to the vertices of finite polyhedra embedded in Euclidean space. In particular he discusses the case of a triangular mesh embedded in Euclidean space. As described in \cite{Ni:2004}, this works for arbitrary oriented $2$-manifolds that have a triangular mesh structure; we follow the review of Banchoff's discrete Morse theory for triangle meshes in \cite{Ni:2004}. 

\begin{definition}
	A \emph{discrete Morse function} $f$ on a triangular mesh $M$ is a real-valued function defined on the vertices of $M$, such that adjacent vertices are mapped to distinct values. 
\end{definition}

Consider a $2$-dimensional oriented triangular mesh $M$, and a discrete Morse function $f$. Extend $f$ to a piecewise-linear function on $M$ by linear interpolation across the edges and faces of $M$. Since $f$ takes different values at the two endpoints of each edge in $M$, the gradient of $f$ will be a non-zero constant across all of the edges and faces of $M$.

The \textit{link} of a vertex $v$ is the set of all vertices $v_1,\ldots,v_m$ that are connected to $v$ by an edge together with the edges that connect $v_i$ and $v_{i+1}$, $1\leq i\leq m$, where we consider $m+1\equiv 1$. Here the order of the vertices is determined by the orientation of the mesh. We denote the link of $v$ by $\Lk(v)$, and the edge connecting vertices $v$ and $w$ by $<v,w>$. The \textit{upper link} of $v$ is defined as the set 
\begin{equation}
{\Lk}^+(v)=\{v_i\in\Lk(v)|f(v_i)>f(v)\}\cup\{\langle v_i,v_j\rangle \in\Lk(v)|f(v_i),f(v_j)>f(v)\}.\nonumber
\end{equation}
Similarly, the \textit{lower link} is defined as
\begin{equation}
{\Lk}^-(v)=\{v_i\in\Lk(v)|f(v_i)<f(v)\}\cup\{\langle v_i,v_j\rangle\in\Lk(v)|f(v_i),f(v_j)<f(v)\}.\nonumber
\end{equation}
Finally, the set of \textit{mixed edges} is defined by
\begin{equation}
{\Lk}^\pm(v)=\{\langle v^+,v^-\rangle\in\Lk(v)|f(v^+)>f(v)>f(v^-)\}.\nonumber
\end{equation}
The link of $v$ decomposes as 
\begin{equation}
\Lk(v)={\Lk}^+(v)\cup{\Lk}^-(v)\cup{\Lk}^\pm(v).\nonumber
\end{equation}
The number of mixed edges is always even and determines the classification of the critical points for $f$.  The classification is shown in the table below. 

\begin{table}[h]
\begin{tabular}{|l|l|l|}
\hline
Type & $|{\Lk}^\pm(v)|$ & Multiplicity\\
\hline
Minimum &$0$ &$-1$ \\
\hline
Maximum &$0$ &$-1$\\
\hline
Regular point &$2$&$0$\\
\hline
Saddle point& $2+2m_v$ &$m_v$\\
\hline
\end{tabular}
\caption{Critical point classification for discrete Morse functions}
\label{critical}
\end{table}

It is customary to call saddle points of multiplicity $1$ \emph{Morse saddles}. 
Banchoff proved in \cite{Banchoff:1970} that the Euler characteristic of the mesh can be computed from the information of the critical points via 
\begin{equation}
\label{eq:EulerCrit}
\chi(M)=\sum_{v\in M}(-m_v),
\end{equation}
the sum going over all vertices of the mesh $M$. 
This formula serves as our motivation, in the next section, to attach a divisor to such a function.

In order for this theory to be applicable to the Riemann surface $X_\mathcal{A}$, we need to endow it with the structure of a triangular mesh. As given, the surface $X_\mathcal{A}$ contains the graph $A$ as a square mesh. To produce a triangulation, we need only subdivide each square into triangles. We cannot do this arbitrarily though --- we need to do this in such a way that we can extend the height function to any vertices that we may have to add. This extended height function does not need to take integer values, but it is still necessary that adjacent vertices take distinct values. 

Coming from the Fuchsian uniformization of our surface discussed in \cite{Doran:2015}, there is a natural triangulation of $X_\mathcal{A}$ induced by the inverse-image of the equator of the Belyi base. 
Recall that the white dots are the pre-images of $0$, the black dots are the pre-images of $1$, and the centers of the faces are the pre-images of $\infty$. 
In order to make use of Banchoff's discrete Morse theory, however, it is necessary that the height function be non-constant along the edges of the triangles. 
This suggests a slight modification to this natural triangulation which we now describe. 

Each face of $X_\mathcal{A}$ can have one of two possible height configurations, the diamond and the bowtie, as depicted below in Figure \ref{Fig:TriDiamond}. If a given face has the diamond configuration, we draw an edge of a new color, grey, that connects the two vertices in the diamond that differ in height (see Figure \ref{Fig:TriDiamond}). If the face has the bowtie configuration, we place a new vertex at the center, which we take to be the pre-image of $\infty$ along the Belyi map, in the face being considered. The added vertex will also be colored grey.  We connect this new vertex to each of the other four vertices in the face by drawing four new grey edges (see Figure \ref{Fig:TriDiamond}). Performing this operation to each face of $X_\mathcal{A}$ produces a triangulation, and the vertices and edges that we needed to draw will all be grey. Denote by $M$ the resulting triangular mesh.

\begin{figure}[h!]
\centering
\subfloat{
	\begin{tikzpicture}[baseline,scale=2]
\draw[green] (0.0441942,-0.4558058) -- (0.4558058,-0.0441942);
\draw[blue] (-0.0441942,-0.4558058) -- (-0.4558058,-0.0441942);
\draw[blue] (0.4558058,0.0441942) -- (0.0441942,0.4558058);
\draw[green] (-0.4558058,0.0441942) -- (-0.0441942,0.4558058);
\draw[gray] (0,-0.4375) -- (0,0.4375);
\draw (0,0.5) circle[radius=0.0625];
\draw (0,-0.5) circle[radius=0.0625];
\draw[fill] (0.5,0) circle[radius=0.0625];
\draw[fill] (-0.5,0) circle[radius=0.0625];
\node[right] at(0,-0.5) {\tiny{$0$}};
\node[above] at(0,0.5) {\tiny{$2$}};
\node[right] at(0.5,0) {\tiny{$1$}};
\node[left] at(-0.5,0) {\tiny{$1$}};
\end{tikzpicture}}
\hspace{15mm}
\subfloat{\begin{tikzpicture}[baseline,scale=2]
\draw[big arrow,green] (0.0441942,-0.4558058) -- (0.4558058,-0.0441942);
\draw[big arrow,blue] (-0.0441942,-0.4558058) -- (-0.4558058,-0.0441942);
\draw[big arrow,blue] (0.4558058,0.0441942) -- (0.0441942,0.4558058);
\draw[big arrow,green] (-0.4558058,0.0441942) -- (-0.0441942,0.4558058);
\draw[big arrow,gray] (0,-0.4375) -- (0,-0.0625);
\draw[big arrow,gray] (0,0.4375) -- (0,0.0625);
\draw[big arrow,gray] (-0.0625,0) -- (-0.4375,0);
\draw[big arrow,gray] (0.0625,0) -- (0.4375,0);
\draw (0,0.5) circle[radius=0.0625];
\draw (0,-0.5) circle[radius=0.0625];
\draw[fill] (0.5,0) circle[radius=0.0625];
\draw[fill] (-0.5,0) circle[radius=0.0625];
\draw[fill,gray] (0,0) circle[radius=0.0625];
\node[right] at(0,-0.5) {\tiny{$0$}};
\node[below right] at(0,0) {\tiny{$\frac{1}{2}$}};
\node[above] at(0,0.5) {\tiny{$0$}};
\node[right] at(0.5,0) {\tiny{$1$}};
\node[left] at(-0.5,0) {\tiny{$1$}};
\end{tikzpicture}
}

\caption{The triangulation of a face with the diamond configuration is depicted on the left; the triangulations of a face with the bowtie configuration is depicted on the right.}
\label{Fig:TriDiamond}
\end{figure}
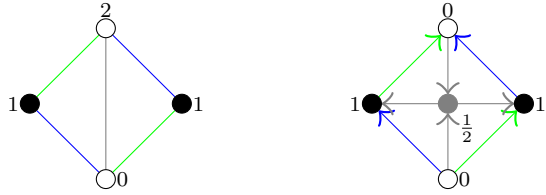

Let $h$ be the height function on the Adinkra $A$. We can extend $h$ to the mesh $M$ by sending each of the additional vertices (at the centers of the bowtie faces) to the average of the heights of the four adjacent vertices. We continue to denote this extension by $h$. By design, $h$ still takes distinct values at adjacent vertices, so $h$ is a discrete Morse function on $M$.

\begin{proposition}
\label{prop:bowcrit}
The vertices at the center of each face with the bowtie configuration are Morse saddle points.
\end{proposition}
\begin{proof}
It is clear from our construction that the set of mixed edges of such a vertex will always contain four elements, whence the result.
\end{proof}
 
\begin{corollary}
All of the maxima and minima of $h$  occur at vertices of the embedded Adinkra.
\end{corollary}

\begin{proposition}
\label{prop:raincrit}

	Fix a vertex $v\in A$. Then, the number of mixed edges is equal to the number of times the value of the height function on the incident vertices of the Adinkra changes as we travel around $v$ counterclockwise. 

\end{proposition}
\begin{proof}
Consider a face incident to the fixed vertex $v$, bounded by edges of color $i$ and $i+1$. Label the vertices incident to $v$ by $v_i,v_{i+1}$ accordingly. There are three cases to consider. First, let us suppose that the face has bow tie configuration. In this case, there is no change in value of the height function from $v_i$ to $v_{i+1}$. The triangular mesh has an additional grey vertex in the center of the face that is incident to all four vertices; its height is equal to the average of the neighbouring heights. It follows that the two grey edges joining $v_i$ and $v_{i+1}$ to the center of the face are not mixed edges.

If the face has the diamond configuration, there are two cases to consider in accordance with whether or not $v$ is the max or min of the face, so that there is no change in the value of the height function from $v_i$ to $v_{i+1}$.  Suppose that $v$ is the max or min, and let $v'$ denote the diagonally opposite vertex that is joined to $v$ by a grey edge. Then, the edges joining $v'$ to $v_i$ and $v_{i+1}$ are not mixed edges. Finally, if $v$ is neither the max nor min, then exactly one of the values of the height function at $v_i$ and $v_{i+1}$ is lower than that $v$, and the other is higher, from which it follows that the grey edge joining $v_i$ and $v_{i+1}$ is a mixed edge. Repeating this argument for each face proves the result.  
\end{proof}

\begin{remark}
	In summary, the height function on Adinkra, embedded in $X_\mathcal{A}$ after choosing a rainbow, can be extended to a piece-wise linear function on $X_\mathcal{A}$ for which the critical behavior satisfies Banchoff's Euler characteristic formula. 
	Propositions \ref{prop:bowcrit} and \ref{prop:raincrit} show that the Adinkra chromotopology together with the height function determine the critical behavior of the mesh $M$.
\end{remark}

\begin{example}[The Valise Adinkra]
\label{ex:valise}
Consider a \emph{valise} Adinkra, for which all of the white vertices have height $0$ and all of the black vertices have height $1$. Each white vertex is a minimum, each black vertex is a maximum, and each face has the bowtie configuration. It follows that  there is a Morse saddle at the center of each face. Since there are $2^{N-k-1}$ vertices of each color and $2^{N-k-2}\cdot N$ faces, we see that $X_\mathcal{A}$ has $2^{N-k-1}$ maxima, $2^{N-k-1}$ minima, and $2^{N-k-2}\cdot N$ Morse saddles. Therefore, by Equation \eqref{eq:EulerCrit},
\begin{equation}
\chi(X_\mathcal{A})=2^{N-k-1}+2^{N-k-1}-N2^{N-k-2}.\nonumber
\end{equation}
Solving $\chi(X_\mathcal{A})=2-2g$ for $g$ shows that $X_\mathcal{A}$ has genus
\begin{equation}\label{genus}
g=1+2^{N-k-3}(N-4),
\end{equation}
for $N\geq 2$. Note that Equation \eqref{genus} agrees with  the genus formula for $X_\mathcal{A}$ found in \cite{Doran:2015}, even though it was derived from Equation \eqref{eq:EulerCrit}, which  only depends on the number of vertices and faces, and not on the number of edges. 

In general it is much easier to find maxima and minima of Adinkras than saddle points. Therefore,  the Euler characteristic equation \eqref{eq:EulerCrit} can be used to constrain the number of saddle points there can be once we compute the genus of $X_\mathcal{A}$ and find the maxima and minima.

\end{example}

\begin{example}[The Fully Extended Hypercube for $N=6$]
Consider the fully extended Adinkra corresponding to the $6$-cube pictured in Figure \ref{Fig:6ExtAdink} with the rainbow (purple, green, light blue, orange, blue, red). There is a single maximum and a single minimum. There are no faces with the bowtie configuration, so all of the saddle points must occur at vertices of the Adinkra. By Proposition \ref{prop:raincrit}, we can determine all of the saddle points from just the Adinkra and the rainbow.

\begin{figure}
\includegraphics[scale=0.25]{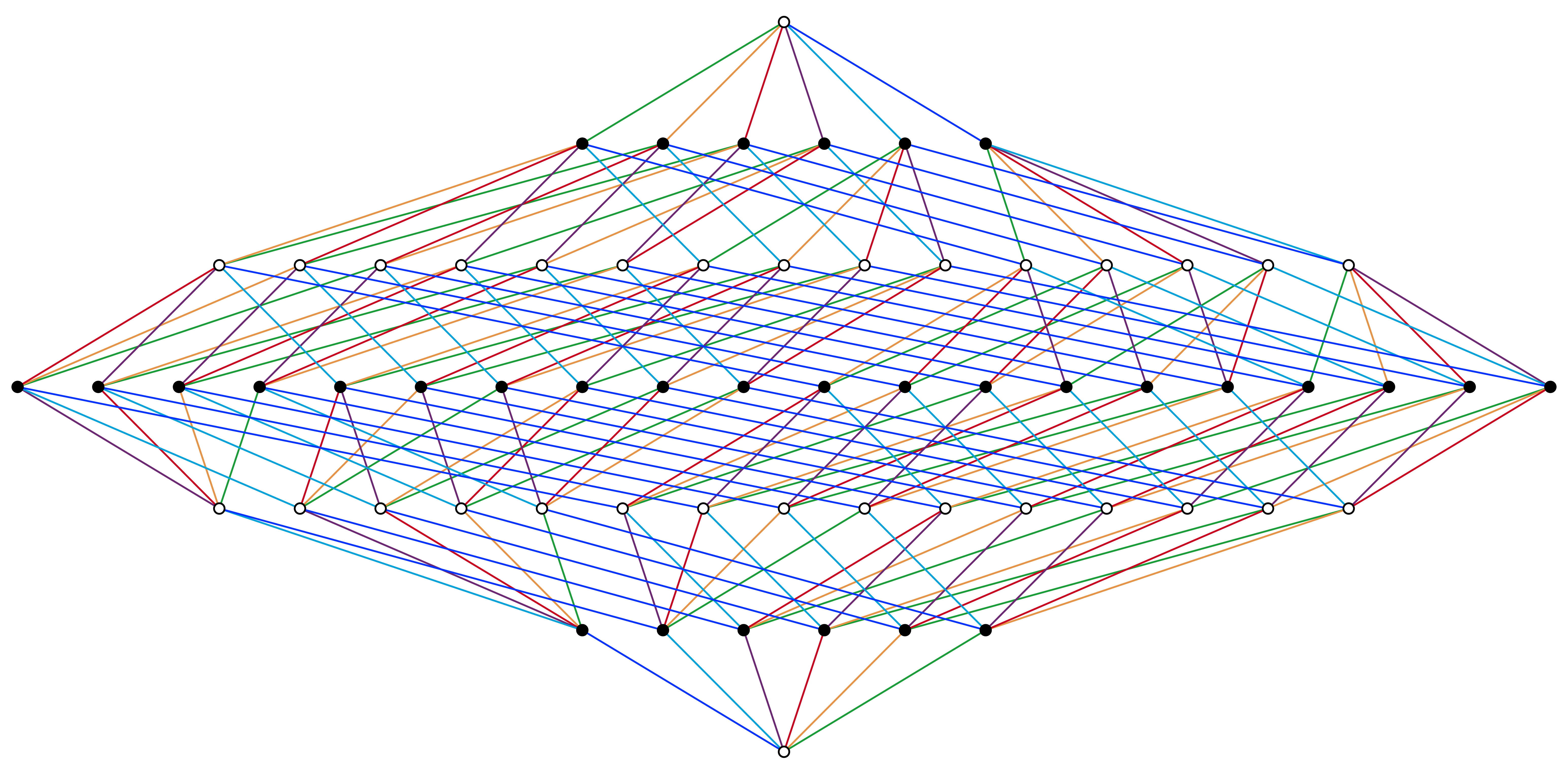}
\caption{The fully extended $6$-cube. All edges point up and the edge dashing is ignored.}
\label{Fig:6ExtAdink}
\end{figure}
A vertex will be regular if and only if all of the edges that point into it (or out of it) are adjacent to each other in the rainbow. The black vertices in the second lowest row are regular since they only have one edge pointing in to them. On the other hand,  the first white vertex on the left in the third lowest row is not regular because the only edges pointing into it are blue and light blue, which are not adjacent in the rainbow. The fourth vertex from the left in the same row \emph{is} regular, since the two incoming edges are blue and orange, which are adjacent.

We can determine the multiplicities of the saddle points. The first white vertex on the left in the third lowest row is a Morse saddle. Indeed, as we cycle through the six edges coming from the vertex according to the rainbow, we change height exactly four times; that is, the number of mixed edges is four, from which it follows that this vertex is a Morse saddle. 

On the other hand, consider the rightmost black vertex in the middle row. As we cycle through the edges incident to the vertex according to the rainbow, we change height a total of six times, from which it follows that this vertex is a saddle of multiplicity two. The only other vertex for which this occurs is the leftmost vertex in the same row. 

Using the genus formula \eqref{genus}, we know that $X_\mathcal{A}$ has genus $17$. We can use this to determine how many Morse saddles there are. We have already found that there is precisely one maximum, one minimum, and two saddles of multiplicity two. There cannot be any saddles of higher multiplicity, so all other saddles must be Morse saddles. Therefore, since  

$$2-2g=-32=1+1-2\cdot(2)-\left(\textrm{number of Morse saddles}\right),$$
it must be the case that there are $30$ such points. One can check that $12$ of these points lie in the middle row, and that each row immediately above or below the middle row contains $9$ saddles.

\end{example}

\begin{definition}
	Let $A$ be an Adinkra and $h$ a height function on it. Let $v$ be a vertex that is a local minimum. The height function $h'$ that agrees on the vertices everywhere except for $v$, at which the value is $h(v)+2$, is called the \emph{vertex raising of} $v$ of the height $h$. Simiarly, if $v$ is a local maximum, we can define the vertex lowering.
\end{definition}

Every Adinkra height function can be obtained from the valise height by raising and lowering vertices \cite{Doran2005}. As we unfold the valise Adinkra, faces with the bowtie topology become diamonds. The Morse saddle points in the valise Adinkra are forced to move to the vertices of the Adinkra as the bowtie faces become undone. Eventually, as in the previous example, some of the Morse saddle points merge to create saddle points of higher multiplicity. With some work, we can determine precisely how the critical behavior is affected when raising or lower a vertex. 

\begin{proposition}\label{raisingprop}
	Let $A\ins X_\mathcal{A}$ be an Adinkra embedded in its associated Riemann surface, and let $h$ be a height function. Assume that $v$ is a minimum (or maximum) for $h$, and let $h'$ be the height function obtained by raising (or lowering) the vertex $v$. Let $v_i\in A$ be the adjacent vertex by the edge of color $i$, and let $f_i$ denote the center of the incident $i/i+1$ colored face. Let $m_p$ and $m'_p$ denote the multiplicities of the point $p$ with respect the height functions $h$ and $h'$. Then we have
	$$m'_v=m_v=1,$$
	
$$m_{f_i}'=\left\{\begin{array}{ll}
	1&\textrm{ if $m_{f_i}=0$}\\
	0&\textrm{ if $m_{f_i}=1$}
\end{array}\right.$$

$$m_{v_i}'=\left\{\begin{array}{ll}
	m_{v_i}+(-1)^\frac{m_{f_i}+m_{f_{i-1}}}{2}&\textrm{if }m_{f_i}=m_{f_{i-1}}\\
	m_{v_i}&\textrm{otherwise}	
\end{array}\right.$$

In particular, the multiplicity of a point may only increase or decrease by at most $1$.

\end{proposition}

\begin{proof}
For sake of clarity, assume once and for all that $v$ is a minimum, so that we are performing a vertex raise. 
Since $v$ started out as a minimum, $m_v=1$; after raising $v$, it must be a maximum for the new height function, so we have $m'_v=1$ as well. This proves the first statement. 

Next, we observe that raising a vertex has the effect of switching the configuration of all incident faces between the bow tie and diamond configuration. 
Since we have seen that the centers of faces with the diamond configuration are regular, and therefore have multiplicity $0$, and that the centers of faces with the bow tie configuration are Morse saddles, with multiplicity $1$, the second statement follows. 

For the third statement, consider the vertex $v_i$. 
The faces $f_i$ and $f_{i-1}$ share the edge of color $i$ that joins $v$ to $v_i$. There are three cases to consider depending on the nature of these adjacent faces. If both faces have the bow tie configuration, then it easy to check that the number of mixed edges around $v_i$ will go up by $2$ when the vertex $v$ is raised. Similarly, if both faces have the diamond configuration, then the number of mixed edges around $v_i$ will go down by $2$. Finally, if there is one face each of each configuration, then the configurations are exchanged after raising, and the number of mixed edges around $v_i$ remains the same.  

Finally, observe that none of the points on the mesh other than $v$, the vertices $v_i$, and the centers $f_i$, are effected by raising $v$, and so no other multiplicity will be affected. 
\end{proof}


\subsection{The Morse Divisors}
In the previous section, we have seen how the height function on an Adinkra gives rise to a discrete Morse function, and how the Adinkra chromotopology used to construct $X_\mathcal{A}$ determines the critical behavior of this discrete Morse function. 
Our goal in this section is to capture this data in a way that is natural for the Riemann surface/algebraic curve $X_\mathcal{A}$. 
To accomplish this, we will associate a divisor on $X_\mathcal{A}$ to each height function.

\begin{definition}
	The \emph{divisor group} of a Riemann surface $X$, $\textrm{Div}(X)$ is the free abelian group on the set of points of $X$; elements of $\textrm{Div}(X)$ are called \emph{divisors}. 	
	The \emph{degree} of the divisor $\sum a_i\cdot P_i$ is the integer $\sum a_i$. 
\end{definition}

Thus a divisor is defined by attaching integer values to a finite number of points on the Riemann surface. 
Inspired by the Euler characteristic formula \eqref{eq:EulerCrit} in the previous section, which stated that
$$\chi(X_\mathcal{A})=\sum_{v\in M}(-m_v),$$
we make the following definition.

\begin{definition}
	Let $\mathcal{A}$ be an Adinkra chromotopology with a height function $h$ and fix some rainbow. Let $M$ be the associated triangulation described in the previous section, and continue to denote by $h$ the extension to $M$ and $m_v$ the multiplicities of the critical points for $h$ (as determined in Table \ref{critical}). The \emph{Morse divisor} of $h$ on $X_\mathcal{A}$ is
	$$D_h:=\sum_{v\in M}(-m_v)\cdot v$$
\end{definition}

\begin{remark}\label{MorseRemark}
	By definition, the Morse divisor $D_h$ is a divisor on $X_\mathcal{A}$ whose support is contained in the pre-image under the Belyi map of the set $\{0,1,\infty\}$, and whose degree is always equal to the Euler characteristic of $X_\mathcal{A}$.  
\end{remark}

\begin{remark}
	Morse divisors encode geometrically the topological complexity of a height function on an Adinkra.
It is clear that the Morse divisor is insensitive to an overall shift in the height function, as well as  the ``flipping'' of the height function which corresponds to \emph{automorphic duality} \cite{Doran2005} since the definition depends only on the critical behavior of the height function and not on the specific values of the height function. 
It is perhaps worth noting that this suggests a new equivalence relation on the set of height functions motivated by topology, rather than physics. 
Namely, we may call two height functions \emph{topologically equivalent} if they give rise to the  same Morse divisor.  
\end{remark}

The following proposition explicitly describes how the Morse divisor is affected when raising or lowering a vertex. 
It codifies the intuition that raising and lowering vertices is a local operation on the graph and the associated Riemann surface.
Specifically, the change in Morse divisor is by a degree $0$ divisor with coefficients in $\{0,\pm 1\}$ supported on a star neighbourhood of the vertex being raised.

\begin{proposition}\label{DivRaise}
Let $A\ins X_\mathcal{A}$ be an Adinkra embedded in its associated Riemann surface, and let $h$ be a height function. Assume that $v$ is a minimum (or maximum) for $h$, and let $h'$ be the height function obtained by raising (or lowering) the vertex $v$. Let $v_i\in A$ be the adjacent vertex by the edge of color $i$, and let $f_i$ denote the center of the incident $i/i+1$ colored face. Let $m_p$ and $m'_p$ denote the multiplicities of the point $p$ with respect the height functions $h$ and $h'$. If $D_h$ and $D_{h'}$ are the associated Morse divisors, then we have 
$$D_{h'}-D_h=\sum_{p\in\textrm{Star}(v)}a_p\cdot p=\sum_{p\in\textrm{Star}(v)-\{v\}}a_p\cdot p,$$
where $\textrm{Star}(v)$ is the star neighbourhood of $v$ consisting of the $N$ adjacent vertices and $N$ centers of the faces incident to $v$, and where the $a_p$ are determined as follows. 
We have $a_v=0$, 

$$a_{f_i}=\left\{\begin{array}{ll}
	0&\textrm{ if $m_{f_i}=0$}\\
	-1&\textrm{ if $m_{f_i}=1$}
\end{array}\right.$$

$$a_{v_i}=\left\{\begin{array}{ll}
	(-1)^\frac{m_{f_i}+m_{f_{i-1}}}{2}&\textrm{if }m_{f_i}=m_{f_{i-1}}\\
	0&\textrm{otherwise}	
\end{array}\right..$$
The difference $D_{h'}-D_h$ is therefore a degree-$0$ divisor supported on $\textrm{Star}(v)$ with coefficients in $\{0,\pm 1\}$.

\end{proposition} 

\begin{proof}
This is a direct consequence of Proposition \ref{raisingprop}, the definition of a Morse divisor, and the Euler characteristic formula. 	
\end{proof}

Naturally, we are left to answer the question of what kind of divisors  are obtained from this construction. 
Before attempting to answer this question, we recall a notion of equivalence that is used in Riemann surface theory, as well as a stratification of the divisor group.

\begin{definition}
	 A divisor is \emph{principal} if it is the divisor of zeroes and poles of a rational function. Two divisors are \emph{linearly equivalent} if their difference is principal. The quotient of $\textrm{Div}(X)$ by the subgroup of principal divisors is the \emph{divisor class group}, or the \emph{Picard} group, denote $\textrm{Pic}(X)$. 
	
Principal divisors have degree $0$, and so the degree map is well defined on $\textrm{Pic}(X)$. The subgroup of $\textrm{Div}(X)$, resp. $\textrm{Pic}(X)$ consisting of divisors, resp. classes of divisors, of degree $0$ is denoted by $\textrm{Div}^0(X)$, resp. $\textrm{Pic}^0(X)$. 
The group $\textrm{Pic}^0(X)$ has the structure of an algebraic variety and is known as the \emph{Jacobian} of $X$. 
More generally, the subset of divisors, resp. classes, of fixed degree $n$ is denoted by $\textrm{Div}^n(X)$, resp. $\textrm{Pic}^n(X)$. 
\end{definition}

By design, the degree of any Morse divisor is equal to the Euler characteristic $\chi$ of the Riemann surface $X_\mathcal{A}$.
It follows that all of our divisors are contained in the subset $\textrm{Div}^\chi(X)$. 
Descending to linear equivalence classes, we may consider the Morse divisor classes in $\textrm{Pic}^\chi(X)$. 
Sitting inside $\textrm{Pic}^\chi(X)$ is a distinguished divisor class known as the anti-canonical divisor class. 
The following proposition and corollary shows that the anti-canonical class can be realized as the Morse divisor class associated to the valise height function. 

\begin{proposition}
	Let $\beta\colon X_\mathcal{A}\to\mathbf{P}^1(\mathbf{C})$ be the Belyi map associated to the Riemann surface $X_\mathcal{A}$. Let $V$ be the set of vertices on the Adinkra, i.e., pre-images of $0$ and $1$, and let $F$ denote the set of centers of faces, i.e., pre-images of $\infty$ under $\beta$. Then, the canonical divisor $K_X$ is given by the following formula:
	$$K_{X_\mathcal{A}}=D_F-D_V,$$ where $$D_V=\sum_{v\in V}v,\ D_F=\sum_{f\in F}f.$$
\end{proposition}

\begin{proof}
The canonical divisor on $\mathbf{P}^1(\mathbf{C})$ is well-known to be $-2(\infty)$. The canonical divisor of $X_\mathcal{A}$ is then equal to the sum of the pull-back of this divisor and the ramification divisor of $\beta$. The pull-back of $-2(\infty)$ under $\beta$ is
$$\beta^*(-2(\infty))=-2\beta^*(-\infty)=-2\cdot D_F,$$ and the ramification divisor is 
$$R_\beta=(N-1)\cdot D_V+D_F,$$ recalling that $\beta$ is totally ramified to order $N$ at the vertices and to order $2$ at the centers of the faces. 

It follows that
$$K_{X_\mathcal{A}}=(N-1)\cdot D_V-D_F.$$
Using the Belyi map $\beta$, which is, by definition, a rational function on $X_\mathcal{A}$, we see that $$N\cdot D_V\equiv N\cdot D_B+N\cdot D_W\equiv D_F+D_F=2\cdot D_F,$$ where $D_B$ and $D_W$ are the formal sums of the black and white vertices. Replacing $N\cdot D_V$ with $2\cdot D_F$ in the above formula shows that
$$K_{X_\mathcal{A}}\equiv D_F-D_V.$$
\end{proof}

\begin{corollary}
	Let $\mathcal{A}$ be a chromotopology and let $D$ be the Morse divisor on $X_\mathcal{A}$ associated to the valise height function after having chosen a rainbow. Then $D$ is equivalent to the anti-canonical divisor. 
\end{corollary}
\begin{proof}
	This follows immediately from our analysis in Example \ref{ex:valise}.
\end{proof}

By using the anti-canonical class as a base point in $\textrm{Pic}^\chi(X_\mathcal{A})$, we obtain a bijection between $\textrm{Pic}^\chi(X_\mathcal{A})$ and $\textrm{Pic}^0(X_\mathcal{A})$ obtained by subtracting the base point from each class in $\textrm{Pic}^\chi(X_\mathcal{A})$. 
The advantage to doing so is that we may use the group structure of $\textrm{Pic}^0(X_\mathcal{A})$ to help us understand $\textrm{Pic}^\chi(X_\mathcal{A})$ --- this is completely analogous to the odd dashings and spin structures studied in Section 2 that were affine spaces over certain cohomology groups.
It follows that an understanding of the Morse divisor classes depends on an understanding of the Jacobian $\textrm{Pic}^0(X_\mathcal{A})$.  
In the next section, we will describe the Jacobian of $X_\mathcal{A}$ and this will allow us to make some statements about the structure of the Morse divisor classes.

\section{Jacobians}

\subsection{The Jacobian of $X_\mathcal{A}$}
After recalling some basic facts about the theory of Abelian varieties and Jacobians of Riemann surfaces, we will describe the Jacobians of the Riemann surfaces coming from Adinkras. 
As remarked in the previous section, the group $\textrm{Pic}^0(X)$ can actually be given the structure of an algebraic variety  and is an example of an \emph{Abelian variety} --- a projective algebraic variety with a group structure. 
The only one-dimensional Abelian varieties are elliptic curves and, in this way, one may view the theory of Abelian varieties as a generalization of elliptic curves to higher dimensional varieties. 

\begin{definition} A \emph{morphism} of Abelian varieties is a morphism as algebraic varieties that respects the group operations. 
	Given two abelian varieties $A_1$ and $A_2$, a surjective map $A_1\to A_2$ is called an \emph{isogeny} if the kernel is finite.
	If such a map exists, we say that $A_1$ and $A_2$ are \emph{isogenous}. 
	
	An abelian variety $A$ is called \emph{decomposable} if it is isogenous to the product of abelian varieties of smaller dimension, and is called \emph{simple} otherwise. 
	Finally, it is called \emph{completely reducible} if it is isogenous to a product of elliptic curves. 
\end{definition}

A question in algebraic geometry that has been well-studied is the \emph{isogenous decomposition} of an abelian variety. Poincar\'{e} has shown the following reducibility theorem:
\begin{theorem}
	If $A$ is an abelian variety, then there exist simple abelian varieties $A_i$ and positive integers $n_i$ such that $A$ is isogenous to the product
	$$A\cong_{isog}A_1^{n_1}\times\cdots \times A_r^{n_r}.$$
	The factors $A_i$ and $n_i$ are unique up to isogeny and permutation of the factors. 
	The decomposition $A\cong_{isog}A_1^{n_1}\times\cdots\times A_r^{n_r}$ is called the  \emph{isogenous decomposition} of $A$. 

\end{theorem}

Using the work of \cite{Hidalgo2015}, we are able to determine the isogenous decompositions of the Jacobians of the Riemann surface $X_\mathcal{A}$. Let us fix $N$ and a choice of rainbow, and consider first the Riemann surface $X_N$ associated to the hypercube. Recall that $X_N$ admits an algebraic model as a complete intersection of quadrics in projective space \cite{Doran:2015}. Namely, we have 
$$X_N:\ \left\{\begin{array}{ccc} x_1^2+x_2^2+x_3^2&=&0\\
\mu_3 x_1^2+x_2^2+x_4^2&=&0\\
\vdots &\vdots &\vdots\\
\mu_{N-1}x_1^2+x_2^2+x_{N}^2&=&0
\end{array}\right.,$$
where the $\mu_i$ are the images of the $N$-th roots of $-1$ under a M\"{o}bius transformation. The map $\pi\colon X_N\to\mathbf{P}^1(\mathbf{C})$ given by 
$$\pi([x_1,\dots, x_N])=[x_1^2,x_2^2]$$
is the realization of the map from $X_N$ to the beachball $B_N$. In this model, the maximal even code $\mathcal{C}$ acts on $X_N$ via sign changes. More explicitly, the codeword $c_i$ with $1$'s in the $i$ and $i+1$ position acts by switching the coordinate $x_i$. 

Given a subgroup of $K\ins \mathcal{C}$, we may consider the induced branched cover $$X_N\to X_N/K,$$ by which we mean the map on the underlying Riemann surface (the quotient will in general be an orbifold). Notice that if $K$ is a doubly-even code, then we simply obtain the Riemann surface corresponding to the associated Adinkra. These maps induce morphisms of Jacobians and, as explained in \cite{Hidalgo2015}, can be used to find the isogenous decomposition of $X_N$. The following proposition is a direct application of Theorem 4.4 in \cite{Hidalgo2015}.

\begin{proposition}\label{jacobiandecom}
	The Jacobian of $X_N$ is isogenous to the product of the Jacobians of the curves $X_N/K$ where $K$ runs over all subgroups of $\mathcal{C}$ of index $2$ such that the quotient has strictly positive genus. Each such quotient is a hyperelliptic curve branched over a subset of the $N$-th roots of $-1$ of even cardinality containing at least four elements, and conversely: each such hyperelliptic curve arises from a subgroup of index $2$ in $\mathcal{C}$. 
\end{proposition}

\begin{corollary}\label{Cor:Paul}
	Let $C\ins\mathcal{C}$ be a doubly even code. If $\mathcal{A}$ is the corresponding Adinkra with rainbow inherited from the hypercube, then the isogenous decomposition of $JX_{\mathcal{A}}$ is given by the product of the Jacobians of the curves $X_N/K$ where $K$ runs over the subgroups of $\mathcal{C}$ of index $2$ that contain the code $C$. 
\end{corollary}

The proof of this corollary is not difficult, but breaks up the exposition. 
The interested reader can find the proof in the Appendix B below. 
\subsection{Examples}
In this section, we will use Proposition \ref{jacobiandecom} to compute the isogenous decompositions of $X_\mathcal{A}$ for some low values of $N$. 
In order to calculate the Jacobians of the quotients, it is convenient to describe the index $2$-subgroups of $\mathcal{C}$ as the kernels of homomorphisms $\phi\colon\mathcal{C}\to\mathbf{Z}/2\mathbf{Z}$. As long as $\phi$ is non-trivial, the kernel will be a subgroup of index $2$. Denote by $X_{N,\phi}$ the quotient of $X_N$ by the kernel of $\phi$. The induced map 
$$X_{N,\phi}\to B_N$$
is a double cover branched over the set of $\mu_i$ for which $c_i\notin\ker\phi$. Let  $A_\phi$ be the set of generators $c_i$ for which $c_i\notin\ker\phi$. Notice that this set will always have even cardinality. The curve $X_{N,\phi}$ will have positive genus if and only if this set has at least $4$ elements. Finally, the corresponding subgroup contains a doubly even code $C$ if and only if the code vanishes under $\phi$. 

\begin{example}[$N\leq 4$]
For $N\leq 3$, the Riemann surfaces all have genus zero, so there is no Jacobian to consider. For $N=4$, both surfaces are elliptic curves after choosing a base point, and are therefore isomorphic to their Jacobians. It is still worthwhile to examine what the above theorem says in detail in this example. 

There is only one subset of the fourth roots of $-1$ that contains at least four elements, namely, the entire set. According to the above theorem, the Jacobian of $X_4$ is therefore isogenous to the Jacobian of the elliptic curve expressed as the double cover of $B_4$ over the $4$-th roots of $-1$. A model is given by 
$$y^2=x^4+1.$$

For $N=4$, the code $\mathcal{C}$ is generated by $c_1,c_2,c_3$, and we have the relation $c_4=c_1+c_2+c_3$. A homomorphism $\phi\colon\mathcal{C}\to\mathbf{Z}/2\mathbf{Z}$ is determined by where $c_1,c_2,c_3$ are mapped to. The homomorphism $\phi$ corresponding the curve seen above is the map determined by sending all of the $c_i$ to $1$. Note that by sending $c_1,c_2,c_3$ to $1$, we must also send $c_4$ to $1$. The corresponding index $2$-subgroup is
$$K=\{(0000),(1010),(0101),(1111)\}.$$

It is evident that the subgroup $K$ contains the unique doubly even code in $N=4$, and so the associated Adinkra surface covers the same elliptic curve. In fact, the elliptic curves $X_4$, $X_{(4,1)}$ and $X_4/K$ are all isomorphic to the same elliptic curve, and by choosing appropriate coordinates, the map from $X_4\to X_4/K$ is simply given by multiplication-by-$2$, and the intermediate map is given by multiplication-by-$(1+i)$, noting that the elliptic curve has CM type by $\mathbf{Z}[i]$.  
	
\end{example}

\begin{example}[$N=5$]

For $N=5$, there are five ways of choosing four of the $5$-th roots of $-1$, leading to five elliptic curves that make up the Jacobian of $X_5$. The five elliptic curves are isomorphic to each other because each of the five choices of branch loci are related by M\"{o}bius transformations corresponding to a rotation by multiples of $\frac{2\pi}{5}$. Therefore, we have
$$JX_5\cong_{isog}E^5,$$ where a model for $E$ can be given as  $$E:\ y^2=x^4-x^3+x^2-x+1.$$

The corresponding subgroups of index $2$ in the maximal even code $\mathcal{C}$ are determined by the five homomorphisms $\phi_i$ defined by sending $c_i$ to $0$ and $c_j$ to $1$ if $j\neq i$. One can check that for each of the five doubly even codes in $N=5$, there are only three homomorphisms $\phi_i$ that vanish on the code. 

For example, let $C=\gen{(11110)}=\gen{c_1+c_3}$. Then $C$ is contained in the kernel of $\phi_2,\phi_4,\phi_5$. The Jacobian of the corresponding Riemann surface is therefore isogenous the product of three copies of $E$. In particular, all of these quotient Riemann surfaces have the same isogenous decomposition of Jacobians. 
	
\end{example}

\begin{example}[$N=6$]
	For $N=6$, the curve $X_6$ has genus $17$, so already the dimension of the Jacobian is quite large. 
	There is a unique hyperelliptic curve $H$ corresponding to choosing all of the $6$-th roots of $-1$, and a model is given by 
	$$H\colon y^2=x^6+1.$$
The Jacobian of $H$ is a $2$-dimensional abelian variety, and one can ask if it is simple or not. 
Notice that $H$ possesses the involution $(-x,y)$ that is distinct from the hyperelliptic involution. 
As explained in \cite{Hidalgo2015}, this can be used to show that the Jacobian of $H$ is isogenous to the product of two elliptic curves. 
In fact, the two elliptic curve factors are isomorphic to 
$$E_1:\ y^2=x^3+1.$$
We remark that $E_1$ has CM type by $\mathbf{Z}[\omega]$ where $\omega$ is a $3$-rd root of $1$. 

In addition to this factor of the Jacobian, there are $15$ ways of choosing $4$ of the $6$-th roots of $-1$ giving rise to $15$ elliptic curves that make up the rest of the Jacobian, up to isogeny. Up to isomorphism, only three elliptic curves appear, as can be deduced by considering the action of the order six rotation group on the set of configurations of four of the six branch points. One can check that none of these elliptic curves have CM type. 

Consider the permutation-equivalent doubly even codes $C_1=\gen{(111100)}$ and $C_2=\gen{(101011)}$. Both of these codes give rise to genus nine Riemann surfaces. 
Of the the sixteen index-$2$ subgroups of $\mathcal{C}$ giving rise to positive genus curves, eight of them contain $C_1$ --- in particular, the subgroup corresponding to the hyperelliptic curve $H$, described above, contains $C_1$. 
On the other hand, there are nine subgroups containing $C_2$, corresponding to nine of the other elliptic curves appearing in the Jacobian decomposition of $X_6$. 
In particular, the curve associated to $C_1$ has a CM type factor appearing in its Jacobian, while the curve associated to $C_2$ does not. 
\end{example}

\begin{remark}
	The above example shows that two surfaces obtained from permutation-equivalent codes may give rise to non-isomorphism Riemann surfaces, despite our claim in \cite{Doran:2015}. In Appendix A, we present a correct and reformulated description of the action of $R$-symmetry on the Riemann surfaces $X_\mathcal{A}$. 
\end{remark}

\begin{example}[The $E_8$ Code]

In this example, we will examine the $E_8$ code, which is the unique doubly even code having maximal dimension $4$ in $\mathbf{F}_2^8$, up to permutation equivalence. The standard $E_8$ codes is given by the following generator matrix:
$$\begin{bmatrix}
	1&1&1&1&0&0&0&0\\
	0&0&1&1&1&1&0&0\\
	0&0&0&0&1&1&1&1\\
	1&0&1&0&1&0&1&0
\end{bmatrix}$$

As noted above, the Riemann surfaces associated to different codes in the same permutation class may be non-isomorphic. 
The permutation equivalence class containing the $E_8$ code consists of $30$ codes. 
The permutation $\eta=(12345678)$ acts on these codes, and it was remarked earlier that the Riemann surfaces associated to the codes $C$ and $C^\eta$ \emph{do} give rise to isomorphic Riemann surfaces. 
It turns out that there are $6$ $\gen{\eta}$-orbits of the $E_8$ code, with $2$ orbits of sizes $2,4,$ and $8$, respectively. 
Label these orbits $\mathcal{O}_i^j$ where $i$ is the size of the orbit, and $j\in\{1,2\}$. 

Using the same methods as in previous examples, we can determine the structure of the Jacobian for these $6$ Riemann surfaces. It turns out, once again, that \emph{all} six of these Jacobians are completely reducible. In order to determine the isogeny classes of the Jacobians, we calculate the conductor of each elliptic curve factor and recall that two elliptic curves are isogenous exactly when they have the same conductor. 

The results are shown in the following two tables. In Table \ref{IsogenyFactors2}, the decomposition of each Jacobian as a product of elliptic curves is displayed. In Table \ref{EllipticFactor}, the data of each elliptic curve is supplied. The column labeled $j$ lists the $j$ invariant while the column $N$ lists the conductor of the elliptic curve.

\begin{center}
\begin{table}
\begin{tabular}{|c|c|}
\hline
{\bf Orbit}
& 
{\bf Jacobian}

\\
\hline
$\mathcal{O}_2^1$&
$E_1^3\times E_2^2\times E_3^4$
\\
\hline
$\mathcal{O}_2^2$&
$E_1^5\times E_3^2\times E_4^2$
\\
\hline
$\mathcal{O}_4^1$
&
$E_1^3\times E_3^2\times E_4^4$
\\
\hline
$\mathcal{O}_4^2$
&
$E_1\times E_2^6\times E_3^2$
\\
\hline
$\mathcal{O}_8^1$
&
$E_1^3\times E_2^2\times E_3^2\times E_4^2$
\\
\hline
$\mathcal{O}_8^2$
&
$E_1\times E_2^2\times E_4^6$
\\
\hline
\end{tabular}
\caption{Isogeny Factors}\label{IsogenyFactors2}
\end{table}
\end{center}

\begin{center}
\begin{table}
\begin{tabular}{|c|c|c|}
\hline
{\bf Curve}
& 
{\bf $j$}
&{\bf $N$}
\\
\hline
$E_1$&
$2^6\cdot 3^3$
&$2^5$
\\
\hline
$E_2$&
$2^{28}\cdot7^6$
&$2^2\cdot11^2\cdot23^2\cdot26497^2\cdot73609^2$
\\
\hline
$E_3$
&
$2^7$
&
$2^7$
\\
\hline
$E_4$
&
$2^{24}\cdot 7^6$
&
$13^2\cdot 241^2 \cdot 9843913^2$
\\
\hline
\end{tabular}
\caption{Elliptic Factor Data}\label{EllipticFactor}
\end{table}
\end{center}
	
\end{example}

\subsection{Morse Divisor Classes}
Now that we have an understanding of the Jacobians of the surfaces $X_\mathcal{A}$, we can study the Morse divisor classes defined in Section 3.2. 
As we saw in the examples above, when $N\leq 3$, there is not much of a story to tell. 
The Jacobian in this case is just a point, and this corresponds to the fact that for genus $0$ surfaces, the degree of a divisor class completely determines the linear equivalence class. 
Therefore, for $N\leq 3$, \emph{all} of the Morse divisor classes are linearly equivalent to the anti-canonical class. As it turns out, this phenomenon also occurs when $N=4$, as we see in the example below. 

\begin{example}

Let us show that all of the Morse divisors associated to $C=\gen{(1111)}$ are equivalent to the anti-canonical divisor. 
The associated Riemann surface $X_\mathcal{A}$ is an elliptic curve once we choose a base point, and is isomorphic to its Jacobian. 
Let $D_1$ and $D_2$ be two Morse divisors obtained from one another by raising or lowering a single vertex. 
Their difference $D$ is a degree $0$ divisor, and we will argue that it must be principal, i.e., trivial in $\textrm{Pic}^0(X_\mathcal{A})$, showing that they are equivalent. 
Since every height function is obtained from the valise height by a sequence of vertex raisings and  lowerings, this will show that all the Morse divisors are equivalent to the anti-canonical class. 

Suppose that $D_1$ is a Morse divisor corresponding to a height function in which the white vertex $w$ is a local minimum. Let $D_2$ be the Morse divisor associated to the height function obtained by raising $w$. Using $w$ as the origin for the group law on the elliptic curve, the difference of the divisors is a sum of point on the elliptic curve. To further aid in these computations, we will use the fact the Riemann surface is isomorphic to $\mathbf{C}/\mathbf{Z}[i]$. Under this isomorphism, the white vertices are identified with the four $2$-torsion points, the black vertices are identified with the translate of the $2$-torsion by $\frac{1}{4}(1+i)$, and the centers of the faces are identified with the remaining eight $4$-torsion points (see Figure \ref{QuotientExampleheight}).

We argue case-by-case. In light of the edge identifications indicated in Figure \ref{QuotientExampleheight}, we see that there are $0,2$, or $4$ adjacent faces that have the diamond configuration, and each possibility can occur. Once we know the configurations of the four adjacent faces, the height function is completely determined once we pick the value at the last remaining vertex, which corresponds to the (identified) corners in the figure. In each case, we may assign either $0$ or $2$ to this vertex; it follows that there are six cases to consider.

Consider the case for which all four adjacent faces have the diamond configuration, and the remaining white vertex has height $0$. 
Using Proposition \ref{DivRaise}, we see that 
$$D=b_1+b_2+b_3+b_4-(f_1+f_2+f_3+f_4),$$ where the $b_i$ denote the four black vertices, and the $f_i$ denote the centers of the four adjacent faces. 

Using our identification of the curve with $\mathbf{C}/\mathbf{Z}[i]$, we find that the sum of the four black vertices is $0$ in the elliptic curve, and similarly for the four centers of the faces. 
It follows that $D$ is trivial in $\textrm{Pic}^0(X_\mathcal{A})$, showing that $D_1$ and $D_2$ are equivalent to each other. 
The other cases are handled similarly, and are left to the interested reader. 

For the case of the $N=4$ hypercube, the exact same argument works, but we will now be working with $8$-torsion points instead of just $4$-torsion points. 
One argues case-by-case using Proposition \ref{DivRaise} to calculate the possible differences of Morse divisors and argues that they are trivial in $\textrm{Pic}^0(X_\mathcal{A})$.
\end{example}
 
As interesting as the above example is, this phenomenon is not one that holds in general, as we will see in the next example for the $N=5$ hypercube surface. 
Before working out the example, we recall the push-forward morphism on divisor groups. 
Given a map $f\colon X\to Y$ of Riemann surfaces, there is a natural push-forward  homomorphism between divisor groups defined by 
$$f_*(\sum_{P\in X}n_P\cdot P)=\sum_{P\in X}n_Pf(P).$$
This map descends to $\textrm{Pic}^0(X)$ and is a morphism of abelian varieties.

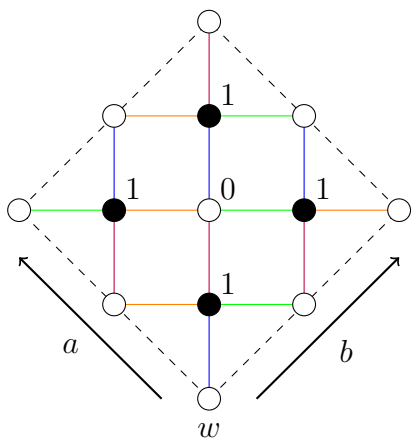
\begin{figure}
\centering
\begin{tikzpicture}[scale=2.5]

	\draw[color=blue] (-0.5,0)--(-0.5,0.5);
	\draw[color=blue] (0,-1)--(0,-0.5);
	\draw[color=blue] (0,0)--(0,0.5);
	\draw[color=blue] (0.5,0)--(0.5,0.5);
	
	\draw[color=purple] (-0.5,-0.5)--(-0.5,0);
	\draw[color=purple] (0,-0.5)--(0,0);
	\draw[color=purple] (0,0.5)--(0,1);
	\draw[color=purple] (0.5,-0.5)--(0.5,0);
	
	\draw[color=green] (-1,0)--(-0.5,0);
	\draw[color=green] (0,-0.5)--(0.5,-0.5);
	\draw[color=green] (0,0)--(0.5,0);
	\draw[color=green] (0,0.5)--(0.5,0.5);

	\draw[color=orange] (-0.5,-0.5)--(0,-0.5);
	\draw[color=orange] (-0.5,0)--(0,0);
	\draw[color=orange] (-0.5,0.5)--(0,0.5);
	\draw[color=orange] (0.5,0)--(1,0);

	\draw[dashed] (0,-1)--(1,0);
	\draw[dashed] (0,-1)--(-1,0);
	\draw[dashed] (1,0)--(0,1);
	\draw[dashed] (-1,0)--(0,1);
	
	\draw[->,thick] (0.25,-1)--node[anchor=north west] {$b$}(1,-0.25);
	\draw[->, thick] (-0.25,-1)--node[anchor=north east] {$a$}(-1,-0.25);

	\draw[fill=white] (-1,0) circle[radius=0.06];
	\draw[fill=white] (0,-1) node[below=5pt] {$w$} circle[radius=0.06];
	\draw[fill=white] (0,0) circle[radius=0.06] node[anchor=south west] {$0$};
	\draw[fill=white] (0,1) circle[radius=0.06];
	\draw[fill=white] (1,0) circle[radius=0.06];
	\draw[fill=white] (-0.5,-0.5) circle[radius=0.06];
	\draw[fill=white] (-0.5,0.5) circle[radius=0.06];
	\draw[fill=white] (0.5,0.5) circle[radius=0.06];
	\draw[fill=white] (0.5,-0.5) circle[radius=0.06];

	\draw[fill=black] (-0.5,0) circle[radius=0.06] node[anchor=south west] {$1$};
	\draw[fill=black] (0,-0.5) circle[radius=0.06] node[anchor=south west] {$1$};
	\draw[fill=black] (0,0.5) circle[radius=0.06] node[anchor=south west] {$1$};
	\draw[fill=black] (0.5,0) circle[radius=0.06] node[anchor=south west] {$1$};

\end{tikzpicture}
\caption{A fundamental domain for the surface corresponding to $C=\gen{(1111)}$. The central white vertex is a local minimum.}
\label{QuotientExampleheight}
\end{figure}

\begin{example}[$N=5$]
Consider the Riemann surface $X_5$. 
Let $D$ be the Morse divisor associated the valise height, and let $D_1$ be obtained by raising a single white vertex. 
We will argue that $D_1$ is not equivalent to $D$ in $\textrm{Pic}^0(X_5)$. 
To do this, we will consider the image of $D-D_1$ under the push-forward map associated to one of the maps $f\colon X\to E$, where $E$ is one of the elliptic curves appearing in the isogenous decomposition of the Jacobian of $X_5$ discussed in the previous section. 

Let $w$ denote the vertex being raised, and let $b_i,f_i$ for $i=1,\dots, 5$ denote the adjacent black vertices and centers of faces to $w$. 
The divisor difference is then
$$D-D_1=b_1+b_2+b_3+b_4+b_5-(f_1+f_2+f_3+f_4+f_5).$$

For concreteness, $\xi=e^{\frac{\pi i}{5}}$, let $\zeta=e^{\frac{2\pi i}{5}}$, and let $\xi_i=\xi\cdot \zeta^{i-1}$. 
Then $\xi_1,\xi_2,\xi_3,\xi_4,\xi_5$ are the $5$-th roots of $-1$, ordered in a counterclockwise fashion. 
The elliptic curve $E$ that we will work with will be branched over the first four roots of $-1$, and so we can take the following equation
$$E\colon\ y^2=(x-\xi_1)(x-\xi_2)(x-\xi_3)(x-\xi_4).$$

In order to work with $E$ computationally, it will be convenient to write it in Weierstrass form. 
To do this, we use a M\"{o}bius transformation to map $\xi_1,\xi_2,\xi_4$ to $0,1,\infty$ respectively. 
We obtain the following different model for $E$:
$$E\colon y^2=x(x-1)(x-(2+\xi^2-\xi^3)).$$

Under the map $X_5\to E$, the points $f_1,f_2,f_3,f_4$ will map to the $2$-torsion points on $E$, and the point $f_5$ will map to a point with $x$-coordinate given $\xi^3-\xi+1$, the image of $\xi_5$ under the M\"{o}bius transformation. 
Using sage \cite{sagemath}, one can verify that the image of $f_5$ is a point of order $4$ in the group law on $E$. 

On the other hand, each of the black vertices must map to one of the two points with $x$-coordinate $-\xi^3-\xi+1$, the image of $\infty$ under the M\"{o}bius transformation. 
These two points are additive inverses of each other, and we will have three of the black vertices mapping to one point, and two black vertices mapping to the other. 
It follows that image of the divisor $b_1+b_2+b_3+b_4+b_5$ is equal to one of these two points. 
Using sage, we see that both of these points have \emph{infinite} order in the group law on $E$. 
Therefore, the push-forward $f_*(D-D_1)$ is the sum of a point of infinite order and a point of order $4$, and is therefore not trivial. 
This shows that $D-D_1$ cannot be principal, and so $D$ and $D_1$ are not linearly equivalent. 
	
\end{example}

While the Morse divisor classes do not always give rise to the anti-canonical class, the above examples illustrate the fact that the set of all Morse divisor classes on a fixed Adinkra naturally give rise to a subset of points lying on the associated Jacobian. Note that this subset of points is not, in general, a subgroup of the Jacobian --- we saw in Example 4.3.2 that we can obtain points of infinite order in the Jacobian despite the fact that there are only finitely many height functions up to shifting, and thus only finitely many points arising from Morse divisor classes. 

\begin{proposition}
	Fix an Adinkra chromotopology $X_\mathcal{A}$, and let $D_v$ denote the Morse divisor associated to the valise height function. The mapping
	$$h\mapsto D_h-D_v$$ induces a well-defined map from the set of height-functions, up to shifting, to the Jacobian of $X_\mathcal{A}$.
\end{proposition}

\begin{proof}
It was remarked earlier that the Morse divisor of a height function is left invariant if we shift the height-function by a constant. 
From the fact that each Morse divisor has the same degree, namely $\chi(X_\mathcal{A})$, it follows that the difference $D_h-D_v$ is a degree-$0$ divisor for every height function $h$. 
The class of this difference in $\textrm{Pic}(X_\mathcal{A})$ is therefore a well-defined point in the Jacobian. 	
\end{proof}

\section{Conclusion}

In this paper we have completed the canonical construction of a geometry associated with the $N$-extended one-dimensional super Poincar\'{e} algebra via Adinkra graphs.  That a naturally defined Riemann surface could be attached to the Adinkra chromotopology as in \cite{Doran:2015} was already surprising.  That the remaining two Adinkra structures, the odd-dashing and the height assignment, also correspond so cleanly to spin structures and discrete Morse functions/divisors, respectively, on this Riemann surface is remarkable.  Recall that, given an Adinkra, the invariant geometric object is an assignment of a Riemann surface (with spin structure and Morse function) to each rainbow.  Since the combinatorial data of the Adinkra graph determines all of these, there is a notion of compatibility between these (possibly nonisomorphic) Riemann surfaces with spin structure.  We have succeeded in showing that the 1D shadows of supermultiplets canonically define a geometric package consisting of collections of super Riemann surfaces with Morse divisors.

It is important to note that its derivation purely from the $N$-extended Poincar\'{e} superalgebra means that this geometry does not {\em a priori} have a pre-determined role to play in any particular supersymmetric physical theory.  Nevertheless, one can reasonably expect geometrized Adinkras to appear in some form in {\em every} context where the requisite amount of supersymmetry is present.  

First evidence for this is provided by the geometrized Adinkra itself.  Both the spin structure and Morse function restrict to Lagrangian (real) curves.  It is striking that this is exactly the data needed to construct the wrapped Fukaya category for the noncommutative homological mirror construction for punctured Riemann surfaces.  A full exploration of the connection between the fundamental building blocks of supersymmetry and the simplest fully-functional model of homological mirror symmetry is underway \cite{Doran2016}.

\section*{Appendix A: $R$-Symmetry Revisited}
Let $C\ins\mathbf{F}_2^N$ be a doubly-even code and let $\sigma\in S_N$ be a permutation that stabilizes the code. 
Fix a rainbow $R$ and denote by $R^\sigma$ the rainbow obtained by applying $\sigma$ to $R$. 
These choices give two chromotopologies $\mathcal{A}$ and $\mathcal{A}^\sigma$; in this appendix we show that that the Riemann surfaces $X_\mathcal{A}$ and $X_{\mathcal{A}^\sigma}$ are isomorphic. 
To prove this fact we will show that the monodromy representations for the two surfaces are globally conjugate. 
 
Since $C^\sigma=C$, the underlying Adinkra graphs are identical and we may therefore label the vertices with the elements of $\mathbf{F}_2^N/C$. 
If $\mathcal{C}\ins\mathbf{F}_2^N$ denotes the maximal even subcode, then the white vertices correspond to the cosets in $\mathcal{C}/C\ins\mathbf{F}_2^N/C$, and the black vertices correspond to the remaining cosets;
in what follows, we will use $W$ and $B$ to denote the set of white and black vertices respectively. 
Denote by $i_c$ the edge of color $i$ incident to the white vertex $w\in W$. 

For simplicity, we assume that the rainbow $R$ is given by $(1,2,\dots, N)$. 
Then, the monodromy over $0$ and $1$ for the Riemann surface $X_\mathcal{A}$ is given by 
$$\sigma_0=\prod_{w\in W}(1_w,\dots, N_w)\ \textrm{and}\ \sigma_1=\prod_{b\in B}(N_{b+e_N},\dots, 2_{b+e_2},1_{d+e_1}),$$
where $e_i$ is the standard basis vector of $\mathbf{F}_2^N$ \cite{Doran:2015}. 
These monodromy permutations are elements of the symmetric group on the set $E:=\{i_w|\ w\in W,\ i=1,\dots N\}$. The monodromy for $X_{\mathcal{A}^\sigma}$ is written down similarly using the rainbow $R^\sigma=(\sigma(1),\dots, \sigma(N))$. 
 We denote the corresponding permutations by $\tilde{\sigma_i}$. 

\begin{Proposition}
	The monodromies for $X_\mathcal{A}$ and $X_{\mathcal{A}^\sigma}$ are globally conjugate. That is, there is a permutation $\tau\in S_E$ satisfying $\tau\sigma_i\tau^{-1}=\tilde{\sigma_i}$. 
\end{Proposition}

\begin{proof}

Let $\phi$ be the automorphism of $\mathbf{F}_2^N$ obtained by permuting the basis vectors $e_i$ in accordance to $\sigma$. 
The map $\phi$ preserves the hamming weight and, by definition, $\phi(C)=C^\sigma$.  

Let $\tau$ be the permutation in $S_E$ given by 
$$\tau\colon i_c\mapsto \sigma(i)_{\phi(c)}.
$$
Then, we have
$$\tau\sigma_0\tau^{-1}=\prod_{w\in W}(\sigma(1)_{\phi(w)},\dots, \sigma(N)_{\phi(w)}).$$
Because the cycles in the above product are disjoint, and can therefore be reordered in the product, it is clear that the right hand side is $\tilde{\sigma_i}$. 

Similarly, we have
$$\tau\sigma_1\tau^{-1}=\prod_{b\in B}(\sigma(N)_{\phi(b+e_N)},\dots, \sigma(2)_{\phi(b+e_2)},\sigma(1)_{\phi(b+e_1)}).$$
Note that $$\sigma(i)_{\phi(b+e_i)}=\sigma(i)_{\phi(b)+\phi(e_i)}=\sigma(i)_{\phi(b)+e_{\sigma(i)}}.$$
It follows from this that $\tau\sigma_1\tau^{-1}=\tilde{\sigma_1}$, and so the monodromies are globally conjugate. 
	
\end{proof}

\section*{Appendix B: Proof of Corollary \ref{Cor:Paul} }

\begin{proof}[Proof of Corollary \ref{Cor:Paul}]
	The result is a consequence of the Kani-Rosen decomposition theorem \cite{Kani1989}. 
	Indeed, if $K_1$ and $K_2$ are two distinct subgroups of index $2$ containing the doubly even code $C$, then $K_1$ and $K_2$ commute with each other since they are subgroups of $\mathcal{C}$, and the group generated by $K_1$ and $K_2$ is all of $\mathcal{C}$. We have $K_1K_2/C =\mathcal{C}/C$ since $K_1K_2 = \mathcal{C}$.  
	In particular, the subgroups $K_i/C$ commute and the quotient by the group generated by them is the beachball $B_N$ of genus $0$. 
	
	It remains to show that the genus of the associated Riemann surface is equal to the sum of the genera of the hyperelliptic subcovers. 
	The following argument is a modification of a proof of this fact found by P. Green. 
	Let $\{c_i\}$ denote the usual generating set for the maximally even code $\mathcal{C}$. 
	As described earlier, the index $2$ subgroups of $\mathcal{C}$ that contain $C$ are in one-to-one correspondence with the non-trivial homomorphisms $\phi\colon\mathcal{C}/C\to\mathbf{F}_2$, of which there are precisely $2^{N-k}-1$, where $k=\dim C$. 
	For each such $\phi$, the hyperelliptic curve $X_\phi$ is branched over the $i$-th roots of $-1$ for which $\phi(c_i)=1$.
	Let $\delta_{\phi,i}$ be $0$ or $1$ (as integers) in accordance with whether $\phi(c_i)$ is $0$ or $1$ (in $\mathbf{F}_2)$. 
	Then, the previous remark implies that
	$$g(X_\phi)=\frac{1}{2}\left(\sum_i \delta_{\phi,i}-2\right).$$
	
	On the other hand, since no $c_i$ lies in the doubly even code $C$, it follows that there are precisely $2^{N-k-1}$ homomorphisms $\phi$ satisfying $\phi(c_i)=1$. 
	That is, 
	$$\sum_{\phi}\delta_{\phi,i}=2^{N-k-2},$$ where the sum is taken over all non-trivial homomorphisms $\phi$. 
	
	Now compute the sum of the genera:
	\begin{eqnarray*}
		\sum_\phi g(X_\phi)&=&\sum_\phi\left(\frac{1}{2}\left(\sum_{i}\delta_{\phi,i}-2\right)\right)\\
		&=&\frac{1}{2}\left(\sum_\phi\sum_i\delta_{\phi,i}\right)-(2^{N-k}-1)\\
		&=&2^{N-k-3}\cdot N-2^{N-k}+1\\
		&=&2^{N-k-3}\cdot(N-4)+1,
	\end{eqnarray*}
	which is precisely the genus of the Riemann surface associated to $C$.
	\end{proof}

\bibliographystyle{hplain}

\end{document}